\documentclass[a4paper]{article}

\usepackage{microtype}
\usepackage{authblk,fullpage}

\title{Bidirectional Nested Weighted Automata}

\author[1]{Krishnendu Chatterjee}
\author[1]{Thomas A. Henzinger}
\author[2]{Jan Otop}

\affil[1]{IST Austria\\
\texttt{\{krish.chat,tah\}@ist.ac.at}}

\affil[2]{University of Wrocław\\
\texttt{jotop@cs.uni.wroc.pl}}

\usepackage{wrapfig}
\usepackage{paralist}
\usepackage{thmtools,thm-restate}
\usepackage{tikz}
\usetikzlibrary{arrows,automata,shapes}

\newcommand{\Paragraph}[1]{\noindent{\textbf{#1}}}
\newcommand{\masterA}{{\cal A}_{mas}}
\newcommand{\nestedA}{\mathbb{A}}
\newcommand{\slaveA}{{\mathfrak{B}}}
\newcommand{\nonnestedA}{{\cal A}}
\newcommand{\silent}[1]{\mathsf{sil}({#1})}
\newcommand{\cost}{{C}}
\newcommand{\masterRun}{\Pi}
\newcommand{\slaveRun}{\pi}
\newcommand{\lang}{{\cal L}}
\newcommand{\valueL}[1]{{\cal L}_{{#1}}}
\newcommand{\abs}{\mathop{\mathsf{Abs}}}

\newcommand{\PTIME}{\textsc{PTime}{}}
\newcommand{\PSPACE}{\textsc{PSpace}{}}
\newcommand{\EXPSPACE}{\textsc{ExpSpace}{}}
\newcommand{\NLOGSPACE}{\textsc{NLogSpace}{}}
\newcommand{\N}{\mathbb{N}}
\newcommand{\Z}{\mathbb{Z}}
\newcommand{\tuple}[1]{\langle #1 \rangle}
\newcommand{\buchi}{B\"{u}chi}
\newcommand{\fsum}{\textsc{Sum}}
\newcommand{\fBsum}[1]{\textsc{Sum}^{L,U}}
\newcommand{\fmax}{\textsc{Max}}
\newcommand{\fmin}{\textsc{Min}}
\newcommand{\flimavg}{\textsc{LimAvg}}
\newcommand{\aut}{{\cal A}}

\newcommand{\Acc}{\mathsf{Acc}}

\newcommand{\InfVal}{\mathsf{InfVal}}
\newcommand{\pad}{\#}
\newcommand{\lpair}[2]{\langle {#1}, {#2} \rangle}
\newcommand{\conf}[1]{conf(#1)}

\usepackage{amsfonts,amsthm,multirow}
\newtheorem{theorem}{Theorem}
\newtheorem{example}[theorem]{Example}
\newtheorem{remark}[theorem]{Remark}
\newtheorem{lemma}[theorem]{Lemma}
\newtheorem{definition}[theorem]{Definition}

\begin{document}

\maketitle

\label{s:intro}

\begin{abstract}
Nested weighted automata (NWA) present a robust and convenient 
automata-theoretic formalism for quantitative specifications.
Previous works have considered NWA that processed input words only in the 
forward direction. 
It is natural to allow the automata to process input words backwards as well, 
for example, to measure the maximal or average time between a response and 
the preceding request.
We therefore introduce and study bidirectional NWA that can process input 
words in both directions.
First, we show that bidirectional NWA can express interesting quantitative 
properties that are not expressible by forward-only NWA. 
Second, for the fundamental decision problems of emptiness and universality, 
we establish decidability and complexity results for the new framework which 
match the best-known results for the special case of forward-only NWA.
Thus, for NWA, the increased expressiveness of bidirectionality is 
achieved at no additional computational complexity.
This is in stark contrast to the unweighted case, where bidirectional finite 
automata are no more expressive but exponentially more succinct than their 
forward-only counterparts.
\end{abstract}

\section{Introduction}

We study an extension of nested weighted automata (NWA)~\cite{nested}  that can process words in both directions. 
We show that this new and natural framework can express many interesting 
quantitative properties that the previous formalism could not.
We establish decidability and complexity results of the basic 
decision problems for the new framework.
We start with the motivation for quantitative properties, 
then describe NWA and our new framework, and finally the contributions.

\smallskip\noindent{\em Weighted automata}.
Automata-theoretic formalisms provide a natural way to express quantitative
properties of systems.
Weighted automata extend finite automata where every transition is
assigned an integer number called weight. 
Thus a run of an automaton gives rise to a sequence of weights.
A value function aggregates the sequence of weights into a single value.
For non-deterministic weighted automata, the value of a word
$w$ is the infimum value of all runs over~$w$.  
First, weighted automata were studied over finite words with weights 
from a semiring, and ring multiplication as value function~\cite{Droste:2009:HWA:1667106},
and later extended to infinite words with limit averaging or supremum as 
value function~\cite{Chatterjee08quantitativelanguages,DBLP:journals/corr/abs-1007-4018,Chatterjee:2009:AWA:1789494.1789497}. 
While weighted automata over semirings can express several 
quantitative properties~\cite{DBLP:journals/jalc/Mohri02}, they cannot
express long-run average properties that weighted automata with limit 
averaging can~\cite{Chatterjee08quantitativelanguages}.
However, even weighted automata with limit averaging cannot express 
some basic quantitative properties (see~\cite{nested}).

\smallskip\noindent{\em Nested weighted automata}.
A natural extension of weighted automata is to add nesting,
which leads to \emph{nested weighted automata (NWA)}~\cite{nested}. 
A nested weighted automaton consists of a master automaton and a set 
of slave automata. The master automaton runs over input infinite words.
At every transition the master can invoke a slave automaton that runs 
over a finite subword of the infinite word, starting at the position where 
the slave automaton is invoked.
Each slave automaton terminates after a finite number of steps and returns 
a value to the master automaton. 
Each slave automaton is equipped with a value function for finite words, 
and the master automaton aggregates the returned values from slave automata 
using a value function for infinite words.

\smallskip\noindent{\em Advantages of NWA}.
We discuss the various advantages of NWA.
\begin{compactenum}
\item For Boolean finite automata, nested automata are equivalent to the 
non-nested counterpart, whereas NWA are strictly more expressive than 
non-nested weighted automata~\cite[Example~5]{nested}.
It has been shown in~\cite{nested} that NWA provide a specification framework 
where many basic quantitative properties can be expressed, which 
cannot be expressed by weighted automata.

\item NWA provide a natural and convenient way to express quantitative 
properties.
Every slave automaton computes a subproperty,
which is then combined using the master automaton.
Thus NWA allow to decompose properties conveniently, and provide a natural 
framework to study quantitative run-time verification.

\item Finally, subclasses of NWA are equivalent in expressive power with 
automata with monitor counters~\cite{nested-sas}, and thus they provide a 
robust framework to express quantitative properties.
\end{compactenum}

\smallskip\noindent{\em Bidirectional NWA.}
Previous works considered slave automata that can only process input 
words in the forward direction (forward-only NWA). 
However, to specify quantitative properties, it is natural to allow slave automata to run 
backwards, for example, to measure the maximal or average time between a response and 
the preceding request.
In this work we consider this natural extension of NWA, namely 
{\em bidirectional NWA}, where slave automata can process words in the forward 
as well as the backward direction.

\smallskip\noindent{\em Natural properties.} 
First, we show that many natural properties can be expressed in the 
bidirectional NWA framework. 
We present two examples below (details in Section~\ref{s:examples}).
\begin{compactenum}
\item {\em Average energy level.} Consider a quantitative setting where each weight 
represents energy gain or consumption, and thus the sum of weights represents 
the energy level. 
To express the average energy level property, the master automaton has 
long-run average as the value function, and at every transition it invokes
a slave automaton that walks backward with sum value function for the weights. 
Thus the average energy level property is naturally expressed by NWA with 
backward-walking slave automata, while  this property is not expressible
by NWA with forward-walking slave automata.

\item {\em Data-consistency property (DCP).} 
Consider the data-consistency property (DCP) where the input letters correspond to 
reads, writes, null instructions, and commits.
For each read, the distance to the previous commit measures how fresh 
is the read with respect to the last commit, and this can be measured with 
a backward-walking slave automaton.
For each write, the distance to the next commit  measures how fresh is 
the write with respect to the following commit, and this can be measured with 
a forward-walking slave automaton.
Thus the average freshness, called DCP,  is expressed 
with bidirectional NWA.
Moreover, the DCP can neither be expressed by NWA with only forward-walking 
slave automata nor by NWA with only backward-walking slave automata.

\end{compactenum}

\smallskip\noindent{\em Our contributions.} 
We propose bidirectional NWA as a specification framework for quantitative 
properties.
First, we show that the classes of forward-only NWA and backward-only NWA 
have incomparable expressiveness, and bidirectional NWA strictly generalize 
both classes. 
Second, we establish complexity of the emptiness and universality problems 
for bidirectional NWA, where we consider the limit-average value function 
for the master automaton and for the slave automata we consider standard 
value functions for finite words (such as min, max, and variants of sum).
The obtained complexity results coincide with the results for forward-only NWA,
and range from $\NLOGSPACE$-complete, $
\PTIME$ to $\PSPACE$-complete to $\EXPSPACE$.
However the proofs for bidirectional NWA are much more involved than forward-only
NWA.
Thus bidirectional NWA have all the advantages of NWA but provide a more expressive 
framework for natural quantitative properties. Moreover, the added expressiveness of 
bidirectionality is achieved with no increase in the computational complexity 
of the decision problems~(Table~\ref{tab:complexity}).
We highlight two significant differences as compared to the unweighted case:
(1)~In the unweighted case bidirectionality does not change expressiveness,
whereas we show for NWA it does; and 
(2)~in the unweighted case for deterministic automata bidirectionality leads to
exponential succinctness and increase in complexity of the decision problems,
whereas for NWA bidirectionality does not change the computational complexity.
Thus the combination of nesting and bidirectionality is very interesting in 
the weighted automata setting, which we study in this work.

\smallskip\noindent{\em Related works.}
Quantitative automata and logic have been extensively studied in recent years
in many different contexts~\cite{Droste:2009:HWA:1667106,Chatterjee08quantitativelanguages,boundsInWRegularity,DBLP:journals/jacm/AlmagorBK16}.
The book~\cite{Droste:2009:HWA:1667106} presents an excellent collection of results
of weighted automata on finite words. 
Weighted automata on infinite words have been studied in~\cite{Chatterjee08quantitativelanguages,DBLP:journals/corr/abs-1007-4018,DrosteR06}.
Weighted automata over finite words extended with monitor counters  have been considered (under the name of 
cost register automata) in~\cite{DBLP:conf/lics/AlurDDRY13,copylessCRA}.
A version of nested weighted automata over finite words has been 
studied in~\cite{bollig2010pebble}, and nested weighted automata over 
infinite words has been studied in~\cite{nested,nestedprob,nested-mfcs}.
Several quantitative logics have also been studied, such as~\cite{BokerCHK14,BouyerMM14,AlmagorBK14}. 
However, none of these works consider the rich and expressive formalism of quantitative
properties expressible by NWA with slaves that walk both forward and backward, 
retaining decidability of the basic decision problems.

In the main paper, we present the key ideas and main intuitions of the proofs of selected results,
and detailed proofs are relegated to the appendix.
 
\section{Definitions}
\label{s:definition}
\newcommand{\run}{\pi}
\newcommand{\blank}{\texttt{\#}}

\subsection{Words and automata}

\Paragraph{Words}.
We consider a finite \emph{alphabet} of letters $\Sigma$.
A \emph{word} over $\Sigma$ is a (finite or infinite) sequence of letters from $\Sigma$.
We denote the $i$-th letter of a word $w$ by $w[i]$, and for $i < j$ we
define $w[i,j]$ as the word $w[i] w[i+1] \ldots w[j]$.
The length of a finite word $w$ is denoted by $|w|$; and the length of an infinite word 
$w$ is $|w| = \infty$.
For an infinite word $w$, word $w[i,\infty]$ is the suffix of $w$ 
with first $i-1$ letters removed.
For a finite word $w$ of length $k$, we define the reverse of $w$, denoted by $w^R$,
as the word $w[k] w[k-1] \ldots w[1]$.

\Paragraph{Labeled automata}. For a set $X$, an \emph{$X$-labeled automaton} $\aut$ is a tuple
$\tuple{\Sigma, Q, Q_0, \delta, F, \cost}$, where
(1)~$\Sigma$ is the alphabet, 
(2)~$Q$ is a finite set of states, 
(3)~$Q_0 \subseteq Q$ is the set of initial states, 
(4)~$\delta \subseteq Q \times \Sigma \times Q$ is a transition relation,
(5)~$F$ is a set of accepting states,
and 
(6)~$\cost : \delta \mapsto X$ is a labeling function.
A labeled automaton $\tuple{\Sigma, Q, \{q_0\}, \delta, F, \cost}$ is 
\emph{deterministic} if and only if 
$\delta$ is a function from $Q \times \Sigma$ into $Q$ 
and $Q_0$ is a singleton. 

\Paragraph{Semantics of (labeled) automata}. 
A \emph{run} $\run$ of a (labeled) automaton $\aut$ on a word $w$ is a sequence of states
of $\aut$ of length $|w|+1$  
such that $\run[0]$ belongs to the initial states of $\aut$
and for every $0 \leq i \leq |w|-1$ we have $(\pi[i], w[i+1], \pi[i+1])$  is a transition of $\aut$.
A run $\pi$ on a finite word $w$ is \emph{accepting} if and only if the last state $\pi[|w|]$ of the run 
is an accepting state of $\aut$.
A run $\pi$ on an infinite word $w$ is \emph{accepting} if and only if some accepting state of $\aut$ occurs
infinitely often in $\pi$. 
For an automaton $\aut$ and a word $w$, we define $\Acc(w)$ as the set of accepting runs on $w$.
Note that for deterministic automata, every word $w$ has at most one accepting run ($|\Acc(w)| \leq 1$).

\Paragraph{Weighted automata and their semantics}.
A \emph{weighted automaton} is a $\Z$-labeled automaton, where $\Z$ is the set of integers. 
The labels are called \emph{weights}. 
We define the semantics of weighted automata in two steps. First, we define the value of a 
run. Second, we define the value of a word based on the values of its runs.
To define values of runs, we will consider  \emph{value functions} $f$ that 
assign real numbers to sequences of integers.
Given a non-empty word $w$, every run $\pi$ of $\aut$ on $w$ defines a sequence of weights 
of successive transitions of $\aut$, i.e., 
$\cost(\pi)=(\cost(\pi[i-1], w[i], \pi[i]))_{1\leq i \leq |w|}$; 
and the value $f(\pi)$ of the run $\pi$ is defined as $f(\cost(\pi))$.
We denote by $(\cost(\pi))[i]$ the weight of the $i$-th transition,
i.e., $\cost(\pi[i-1], w[i], \pi[i])$.
The value of a non-empty word $w$ assigned by the automaton $\aut$, denoted by  $\valueL{\aut}(w)$,
is the infimum of the set of values of all {\em accepting} runs;
i.e., $\inf_{\pi \in \Acc(w)} f(\pi)$, and we have the usual semantics that the infimum of the
empty set is infinite, i.e., the value of a word that has no accepting run is infinite.
Every run $\pi$ on the empty word has length $1$ and the sequence $\cost(\pi)$ is empty, hence 
we define the value $f(\pi)$ as an external (not a real number) value $\bot$. 
Thus, the value of the empty word is either $\bot$, if the empty word is accepted by $\aut$, or $\infty$ 
otherwise.
To indicate a particular value function $f$ that defines the semantics,
we call a weighted automaton $\aut$ with value function $f$ an $f$-automaton. 

\Paragraph{Value functions}.
For finite runs we consider the following classical value functions: for runs of length $n+1$ we have
\begin{compactitem}
\item {\em Max and min:} 
$\fmax(\pi) = \max_{i=1}^n (\cost(\pi))[i]$ and 
$\fmin(\pi) = \min_{i=1}^n (\cost(\pi))[i]$.
\item \emph{Sum and absolute sum:} the sum function 
$\fsum(\pi) = \sum_{i=1}^{n} (\cost(\pi))[i]$, 
the absolute sum 
$\fsum^+(\pi) = \sum_{i=1}^{n} \abs((\cost(\pi))[i])$, where $\abs(x)$  is the absolute value of $x$.
\item \emph{Variants of bounded sum:} we consider a family of functions called the (variant of) bounded sum value function $\fBsum{B}$. 
Each of these functions returns the sum if all the 
partial sums are in the interval $[L,U]$, otherwise there are many possibilities which lead to multiple variants.
For example, we can require that for all prefixes $\pi'$ of $\pi$ we have $\fsum(\pi') \in [L,U]$. We can impose a similar restriction on all suffixes, all infixes etc.
Moreover, if partial sums are not contained in $[L,U]$, a bounded sum can return $\infty$, the first violated bound,  etc.
\end{compactitem}
For infinite runs we consider:
\begin{compactitem}
\item {\em Limit average:} $\flimavg(\pi) = \liminf\limits_{k \rightarrow \infty} \frac{1}{k} \cdot \sum_{i=1}^{k} (\cost(\pi))[i]$.
\end{compactitem}

\Paragraph{Silent moves}. Consider a $(\Z \cup \{ \bot\})$-labeled automaton. We consider such an automaton as an extension
of a weighted automaton in which transitions labeled by $\bot$ are \emph{silent}, i.e., they do not contribute to 
the value of a run. Formally, for every function $f \in \InfVal$ we define
$\silent{f}$ as the value function that applies $f$ on sequences after removing $\bot$ symbols.
The significance of silent moves is as follows: it allows to ignore transitions, and thus provides
robustness where properties could be specified based on desired events rather than steps.

\subsection{Nested weighted automata}
Nested weighted automata (NWA)  have been introduced in~\cite{nested} and originally
allowed slave automata to move only forward. The variant we define here allow two types of slave automata, forward walking and backward walking. 
The original definition of NWA from~\cite{nested} is versatile and hence it
can be seamlessly extended to the case with bidirectional (forward- and backward-walking) slave automata.
We follow the description of~\cite{nested}. 

\smallskip\noindent{\em Informal description.}
A \emph{nested weighted automaton} consists of a labeled automaton over infinite words, 
called the \emph{master automaton}, a value function $f$ for infinite words,  
and a set of weighted automata over finite words, called \emph{slave automata}. 
A nested weighted automaton can be viewed as follows: 
given a word, we consider the run of the master automaton on the word,
but the weight of each transition is determined by dynamically running 
slave automata; and then the value of a run is obtained using the 
value function $f$.
That is, the master automaton proceeds on an input word as an usual automaton, 
except that before taking a transition, it starts a slave automaton 
corresponding to the label of the current transition. 
The slave automaton starts at the current position of the master automaton in the input word
and works on some finite part of it. 
There are two types of slave automata: (a)~forward walking, which move onward the input word (toward higher positions), and
(b)~backward walking, which move towards the beginning of the input word.
Once a slave automaton finishes,
it returns its value to the master automaton, which treats the returned
value as the weight of the current transition that is being executed.
The slave automaton might immediately accept and return value $\bot$,
which corresponds to a \emph{silent transition}, i.e., transition with no weight.
If one of slave automata rejects, the nested weighted automaton rejects.
We present two examples of properties expressible by NWA. 
Additional examples are presented in Section~\ref{s:examples}.

\begin{example}[Average response time and its dual]
\label{ex:ART}
\label{ex:dual-ART}
Consider infinite words over $\{r,g,\#\}$, where $r$ represents 
requests, $g$ represents grants, and $\#$ represents idle. 
A basic and interesting property is the average number of letters
between a request and the corresponding grant, which represents the 
long-run {\em average response time (ART)} of the system.
This property cannot be expressed by a non-nested automaton~\cite{nested}.
ART can be expressed by a deterministic nested weighted automaton, 
which basically implements the definition of ART.
This automaton invokes at every request a forward-walking slave automaton with $\fsum^+$ value function, which counts the number of events until the following grant. 
On the other events the NWA takes silent transitions.
Finally, the master automaton applies $\flimavg$ value function to the values returned by slave automata.
Figure~\ref{fig:ARTAW} presents a run of the NWA computing ART.

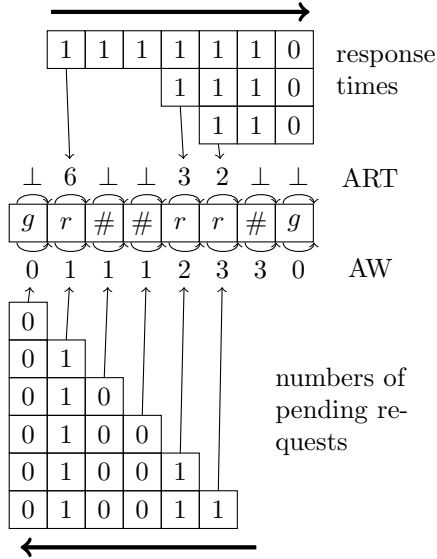
\begin{figure}
\centering
\begin{tikzpicture}

\foreach \i/\l in {1/g,2/r,3/\pad,4/\pad,5/r,6/r,7/\pad,8/g}
{
   \node[rectangle,minimum height=0.5cm, minimum width=0.5cm,draw,inner sep=0pt] (A\i) at (\i*0.5,0) {$\l$};
}

\foreach \x/\y/\w in {1/\bot,2/6,3/\bot,4/\bot,5/3,6/2,7/\bot,8/\bot}
{
\draw[out=90,in=90,->] (A\x.north west)+(0.11,0) to node[above] (B\x) {$\w$} (A\x.north east);
}

\foreach \x/\y/\l in {2/5/1,3/5/1,4/5/1,5/5/1,6/5/1,7/5/1,8/5/0,5/4/1,6/4/1,7/4/1,8/4/0,6/3/1,7/3/1,8/3/0}
{
   \node[rectangle,minimum height=0.5cm, minimum width=0.5cm,draw,inner sep=0pt] at (\x*0.5,\y*0.5-0.2) {$\l$};
}

\foreach \x/\y in {2/5,5/4,6/3}
{
 \draw[->] (\x*0.5,\y*0.5-0.45) to (B\x);
}

\foreach \x/\y/\w in {1/0,2/1,3/1,4/1,5/2,6/3,7/3,8/0}
{
\draw[out=270,in=270,->] (A\x.south west)+(0.11,0) to node[below] (C\x) {$\w$} (A\x.south east);
}

\foreach \x/\y/\l in {1/2/0,2/3/1,3/4/0,4/5/0,5/6/1,6/7/1}
{
  \foreach \i in {\y,...,7}  
  {
   \node[rectangle,minimum height=0.5cm, minimum width=0.5cm,draw,inner sep=0pt] at (\x*0.5,-\i*0.5-0.3) {$\l$};
   }
}

\foreach \x/\y in {1,...,6}
{
 \draw[->] (\x*0.5,-\x*0.5-0.55) to (C\x);
}

\node at (5.0, 0.6) {ART};
\node at (5.0,-0.6) {AW};

\node[text width=1.5cm] at (5.3,2.0) {response times};

\node[text width=2.5cm]  at (5.0,-2.5) {numbers of pending requests};

\draw[ultra thick,->] (3.5,-4.3) to (0.4,-4.3);
\draw[ultra thick,->] (0.8,2.8) to (4.2,2.8);

\end{tikzpicture}
\caption{Runs of NWA computing ART (above) and AW (below). Each weight of a transition is dynamically computed as the sum of weights of slave automata. The thick arrows depict directions of slave automata. }
\label{fig:ARTAW}
\end{figure}

We define the average workload property (AW), which measures
the average number of pending requests. The average is computed over all positions in a word. 
Intuitively, if we pick a position in word $w$ at random, the expected number of pending requests is the average workload of~$w$.
Formally, we define the workload at $i$ in $w$, denoted $wl(w,i)$, as the number of letters $r$ among $w[j,i]$, where $j$ is the last 
position in $w[1,i]$ where $g$ occurs or $1$ if such a position does not exist.
The average workload of $w$ is the limit average over all positions $i$ of $wl(w,i)$.

AW can be expressed by a deterministic $(\flimavg;\fsum^+)$-automaton with backward-walking slave automata. 
Basically, the NWA invokes at every position a slave automaton, which counts the number of $r$ letter from its current position to the first  position containing letter $g$, where it terminates. Since  slave automata run backwards, each of them computes the workload at the position of its invocation.
Figure~\ref{fig:ARTAW} presents a run of the NWA computing AW.
\end{example}

Now, we present a formal definition of NWA and their semantics.

\Paragraph{Nested weighted automata}. 
A \emph{nested weighted automaton} (NWA) with bidirectional slave automata is a tuple $\tuple{\masterA; f; \slaveA_{-m},\ldots, \slaveA_0, \ldots, \slaveA_l}$, with $m,l \in \N$ where
(1)~$\masterA$, called the \emph{master automaton}, is a $\{-m, \ldots, l\}$-labeled automaton over infinite words 
(the labels are the indexes of automata $\slaveA_{-m},  \ldots, \slaveA_l$), 
(2)~$f$ is a value function on infinite words, called the \emph{master value function}, and
(3)~$\slaveA_{-m}, \ldots, \slaveA_l$ are weighted automata over finite words called \emph{slave automata}.
Intuitively, an NWA can be regarded as an $f$-automaton whose weights are dynamically computed at every step by the corresponding slave automaton.
The automata $\slaveA_{-m}, \ldots, \slaveA_{-1}$ (resp., $\slaveA_{1}, \ldots, \slaveA_l$) are called \emph{backward walking} 
(resp., \emph{forward walking}) slave automata.
We refer to NWA with both forward and backward walking slave automata as 
{\em bidirectional NWA}. 
The automaton $\slaveA_0$ immediately accepts and returns no weight; it is used to implement silent transitions, which have no weight.
We define an \emph{$(f;g)$-automaton} as an NWA where the master value function is $f$ and all slave automata are $g$-automata.

\Paragraph{Semantics: runs and values}.
A \emph{run} of $\nestedA$ on an infinite word $w$ is an infinite sequence 
$(\masterRun, \slaveRun_1, \slaveRun_2, \ldots)$ such that 
(1)~$\masterRun$ is a run of $\masterA$ on $w$;
(2)~for every $i>0$ the label $j = \cost(\masterRun[i-1], w[i], \masterRun[i])$ pointers at a slave automaton and 
(a)~if $j < 0$, then $\slaveRun_i$ is a run of the automaton $\slaveA_j$ on some prefix of the reverse word $(w[1,i])^R$, and
(b)~if $j \geq 0$, then $\slaveRun_i$ is a run of the automaton $\slaveA_j$ on some finite prefix of $w[i,\infty]$.
The run $(\masterRun, \slaveRun_1, \slaveRun_2, \ldots)$ is \emph{accepting} if all 
runs $\masterRun, \slaveRun_1,  \slaveRun_2, \ldots$ are accepting (i.e., $\masterRun$ satisfies its acceptance 
condition and each $\slaveRun_1,\slaveRun_2, \ldots$ ends in an accepting state)
and infinitely many runs of slave automata have length greater than $1$ (the master automaton takes infinitely many non-silent transitions).
The value of the run $(\masterRun, \slaveRun_1, \slaveRun_2, \ldots)$ is defined as 
$\silent{f}( v(\pi_1) v(\pi_2) \ldots)$, where $v(\pi_i)$ is the value of the run $\pi_i$ in 
the corresponding slave automaton, and $\silent{f}$ is the value function that takes its input sequence, removes 
$\bot$ symbols and applies $f$ to the remaining sequence.
The value of a word $w$ assigned by the automaton $\nestedA$, denoted by  
$\valueL{\nestedA}(w)$, is the infimum of the set of values of all {\em accepting} runs.
We require accepting runs to contain infinitely many non-silent transitions
as $f$ is a value function over infinite sequences, hence the sequence 
$v(\pi_1) v(\pi_2) \ldots$ with $\bot$ removed must be infinite.

\Paragraph{Deterministic nested weighted automata}. An NWA $\nestedA$ is \emph{deterministic} if (1)~the master automaton 
and all slave automata are deterministic, and (2)~in all slave automata, accepting states have no outgoing transitions. 
Intuitively, a slave automaton in an accepting state can choose (non-deteministically) to terminate or continue running; condition (2) removes this source of non-determinism.

\Paragraph{Width of NWA}. 
An NWA has \emph{width} $k$ if and only if in every  run at every position at most $k$ slave automata are active.

\section{Examples}
\label{s:examples}

In this section we present several examples of quantitative properties that
can be expressed with bidirectional NWA.

\begin{example}[Average energy level]
\label{ex:energy}
We consider the average energy level property studied in~\cite{DBLP:journals/tac/ChatterjeeP15,DBLP:journals/corr/BouyerMRLL15}.
Consider $W \in \N$ and an alphabet $\Sigma_W$ consisting of integers from 
interval $[-W,W]$. 
These letters correspond to the energy change, i.e.,
negative values represent energy consumption whereas positive values 
represent energy gain. 
For $w \in \Sigma_W$ we define the energy level at $i$ as the sum 
$w[1] + \ldots + w[i]$.
The average energy property (AE) is the limit average of the energy levels 
at every position. 
For example, the average energy level of $2 (-1) 3 ((-1) 1)^{\omega}$ is $4$.

AE can be expressed by a $(\flimavg;\fsum)$-automaton $\nestedA$ with 
backward-walking slave automata, but it is not expressible by 
$(\flimavg;\fsum)$-automata with forward-walking slave automata. 
To express AE, a $(\flimavg;\fsum)$-automaton $\nestedA$ with backward-walking 
slave automata invokes at every position a slave automaton, which
runs backward to the beginning of the word and sums up all the letters. 
In contrast, $(\flimavg;\fsum)$-automata with forward-walking slave automata 
 can use finite memory of the master automaton, but finite prefixes influence only finitely many values returned by  slave slave automata and the limit-average value function neglects finite prefixes. 
Formally, we can show with a simple pumping argument that for every  $(\flimavg;\fsum)$-automaton with 
forward-walking slave automata, among words $w_i = 1^i 0^{\omega}$ there exists a pair of words with 
the same value. In contrast, all these words have different AE (AE of $w_i$ is $i$).

AE property is often considered in conjunction with bounds on energy values. 
Typically, energy should not drop below some threshold, in particular, it should not be negative. 
In addition, the energy storage is limited, which motivates the upper bound on the stored energy, where the excess energy is released.  
These two restrictions lead to the \emph{interval constraint} on energy levels, i.e., we require the 
energy level at every position to belong to a given interval $[L,U]$, which results in a 
variant of the bounded sum $\fBsum{L,U}$.
\end{example}

\begin{example}[Data consistency]
\label{ex:data-consistency}
Consider a database server, which processes instructions grouped into transactions. 
There are four instructions: read $r$, write $w$, void $\#$ and commit $c$. 
The commit instruction applies all writes, finishes the current transaction and starts a new one. 
The read instructions refer to writes applied before the previous commit. 

In the presence of multiple clients connected to the database, there are two options to achieve consistency. 
One option is to use locks that can limit concurrency. A second approach is \emph{optimistic concurrency} which proceeds without locks, 
and then rolls back in case there was a collision between transactions.
In ordered to limit the number of roll backs, it is preferred that the read instructions occur shortly after commit, 
while write instructions are followed by the commit instruction as quickly as possible.
Formally, we define (a)~consistency (or freshness) of a read instruction as the number of steps to 
the first preceding commit instruction, and (b)~consistency of a write instruction as the number of steps to the following commit instruction.
The data consistency property (DCP) of $w$ is the limit average of consistency of reads and writes in $w$.

DCP is expressed by the following deterministic $(\flimavg;\fsum^+)$-automaton $\nestedA$ with bidirectional slave automata. 
On every read $r$ (resp., $w$), the NWA $\nestedA$ invokes a slave automaton which walks backward (resp., forward) and counts 
the number of steps to the first encountered $c$.
On the remaining instructions $c,\#$, the NWA $\nestedA$ invokes a dummy slave automaton which corresponds to a silent transition.
\end{example}

\begin{example}
\label{ex:regret}
Consider the framework of Example~\ref{ex:data-consistency}. 
For every position with read $r$ or write $w$ we define a regret at position $i$ as the minimal distance to 
the preceding or the following commit $c$. 
Intuitively, the regret corresponds to the number of instructions by which we have to prepone or postpone 
the commit to include the instruction at the current position. 
We consider the minimal regret property (MR) on words over $\{ r,w,c,\#\}$ defined the limit average over 
positions with $r$ and $w$ of the regret at these positions. 
MR can be expressed by a non-deterministic $(\flimavg;\fsum^+)$-automaton with bidirectional slave automata, 
which basically implements the definition of MR (the non-deterministic guess is whether it is the preceding or
the following commit). 
The NWA invokes at every $r$ or $w$ position one 
of the following two slave automata $\slaveA_B, \slaveA_F$. 
The automaton $\slaveA_B$ counts the number of steps to the preceding grant, while 
$\slaveA_F$ counts the number of steps to the following grant.
\end{example}
 
\section{Decision questions}
\label{s:decision}
For NWA with bidirectional slave automata, we consider the quantitative counterparts of the fundamental problems of emptiness and universality.
The (quantitative) emptiness and universality problems are defined in the same way for weighted automata and all variants of NWA;
in the following definition $\aut$ denotes either a weighted automaton or an NWA.
\smallskip

\Paragraph{Emptiness and universality}.
Given an automaton $\aut$ and a threshold $\lambda$, the {\em emptiness} (resp. {\em universality}) problem asks 
whether there exists a word $w$ with $\lang_\aut(w) \leq \lambda$
(resp., for every word $w$ we have $\lang_\aut(w) \leq \lambda$).
\smallskip

\begin{remark}\label{rem:decision}
The emptiness and universality problems have been studied for forward-only NWA in~\cite{nested}.
\begin{itemize}
\item For NWA the value functions considered for the master automaton are 
the infimum (or limit-infimum), the supremum (or limit-supremum), and the limit-average. 
For all the decidability results for the infimum (limit-infimum) and the supremum (limit-supremum) 
value functions the techniques are similar to unweighted automata~\cite{nested}, which can be easily 
adapted to the bidirectional framework. 
Hence in the sequel we only focus on bidirectional NWA with the limit-average value function for the 
master automaton.

\item Moreover, we study only the emptiness problem for the following reasons.
First, for the deterministic case the emptiness and the universality problems are similar and hence we focus on the emptiness problem.
Second, in the non-deterministic case the universality problem is already undecidable for $\flimavg$-automata even 
with no nesting~\cite{DBLP:journals/corr/abs-1006-1492}.
\end{itemize}
\end{remark}

\subsection{The minimum, maximum and bounded sum value functions}

First, we show that 
for $g$ being $\fmin, \fmax$, or a variant of the bounded sum value function $\fBsum{B}	$, the emptiness problem for $(\flimavg;g)$-automata with bidirectional slave automata is decidable in $\PSPACE$. 
To show that, we prove a stronger result, i.e., every $(\flimavg;g)$-automaton can be effectively transformed to a $\flimavg$-automaton of exponential size.
\smallskip

\noindent\emph{Key ideas}. Weighted automata with value functions $\fmin,\fmax,\fBsum{B}$ are close to (non-weighted) finite-state automata.
In particular, these automata have finite range and for each value $\lambda$ from the range, the set of words of value $\lambda$ is regular. 
Thus, instead of invoking a slave automaton, the master automaton can non-deterministically pick value $\lambda$ and verify that the value returned by this slave automaton is $\lambda$.
For backward-walking slave automata the guessing can be avoided as the master automaton can simulate (the reverse of) runs of all backward-walking slave automata until the current position.
Thus, we can eliminate slave automata from NWA, i.e., we transform such NWA to weighted automata.
Formally, we show that for $g \in \{\fmin,\fmax,\fBsum{B}\}$, 
every $(\flimavg;g)$-automaton with bidirectional slave automata can be transformed into an equivalent $\flimavg$-automaton of exponential size.
The emptiness problem for non-deterministic $\flimavg$-automata is in $\NLOGSPACE$  (assuming weights given in unary) and hence 
we have the containment part in the following Theorem~\ref{th:regularBF}. The hardness part follows from $\PSPACE$-hardness of the emptiness problem for 
 $(\flimavg;g)$-automata with forward-walking slave automata only~\cite{nested}.

\begin{restatable}{theorem}{RegularForwardAndBackward}
Let $g \in \{\fmin,\fmax,\fBsum{B}\}$.
The emptiness problem for non-deterministic $(\flimavg;g)$-automata with bidirectional slave automata is $\PSPACE$-complete.
\label{th:regularBF}
\end{restatable}

\noindent\emph{Note}. The complexity in Theorem~\ref{th:regularBF} does not depend on encoding of weights in slave automata, i.e., the problem is $\PSPACE$-hard even for a fixed set of weights, and
it remains in $\PSPACE$ for weights encoded in binary. 

The average energy property from Example~\ref{ex:energy} with bounds on energy levels can be expressed with $(\flimavg;\fBsum{B})$-automata. 
The emptiness problem for these automata is decidable by Theorem~\ref{th:regularBF}.

\begin{remark}[Parametrized complexity]
 If we assume that the size of slave automata in Theorem~\ref{th:regularBF} is bounded by a constant, the complexity 
of the emptiness problem drops to $\NLOGSPACE$-complete. $\NLOGSPACE$-hardness follows from  $\NLOGSPACE$-hardness of the emptiness problem for
$\flimavg$-automata, which can be considered as a special case of NWA.
\label{rem:parametricRegular}
\end{remark}

The results of this section apply to general bidirectional NWA. 
In the following section we consider bidirectional NWA with the sum value function,
where we consider additional restrictions of finite width (Section~\ref{s:finite}) 
and bounded width (Section~\ref{s:bounded}).
We also justify in Remark~\ref{rem:FinieWidthNatural} that the finite width 
restriction is natural.

\section{Finite-width case}
\label{s:finite}
In this section we study NWA satisfying the \emph{finite width} condition. 
First, we briefly discuss the finite-width condition and argue that it is a natural restriction. 
Next, we show that the emptiness problem for (finite-width) $(\flimavg;\fsum^+)$-automata with bidirectional slave automata is decidable in $\EXPSPACE$.
We conclude this section with the expressiveness results; we show that classical NWA with forward-walking slave automata and NWA with backward-walking slave automaton have incomparable expressive power.
Hence, (finite-width) $(\flimavg;\fsum^+)$-automata with bidirectional slave automata are strictly more expressive than NWA with one-direction slave automata. 

\subsection{The finite-width condition}

\Paragraph{Finite width}. An NWA $\nestedA$ has finite width if and only if in every accepting run of $\nestedA$ at every position at most 
finitely many slave automata are active.
Classical NWA with forward-walking slave automata only have finite width. 
Indeed, in any run, at any position $i$ at most $i$ slave automata can be active.

\begin{example}
Consider an NWA over $\{a,b\}$ such that the master automaton accepts a single word $a b^{\omega}$ and all slave automata are backward walking and accept words $b^* a$.
All slave automata terminate at the first position of $a b^{\omega}$ and hence this NWA does not have finite width.
\end{example}

The automata expressing properties from Examples~\ref{ex:dual-ART},~\ref{ex:data-consistency}~and~\ref{ex:regret} are 
 finite-width $(\flimavg;\fsum^+)$-automata with bidirectional slave automata. 
Observe that an NWA does not have finite width if and only if it has an accepting run, in which at some position $i$ infinitely many backward-walking slave automata terminate.

\begin{remark}[Finite width is natural for positive sum]
\label{rem:FinieWidthNatural}
Let $\nestedA$ be a $(\flimavg;\fsum^+)$-automaton with bidirectional slave automata.
Except for degenerate cases, runs of $\nestedA$, which do not have finite width, have infinite value. Indeed, 
consider a run $\pi$ and a position $i_0$ at which infinitely many automata are active. Since only finitely many forward-walking slave automata are active at $i_0$, infinitely many of them are backward-walking and 
for some position $i < i_0$, infinitely many slave automata $S$ terminate at position $i$.
 Then, one of the following holds: either that value of this run is
infinite or one of the following two degenerate cases happen:
(a)~The slave automata from $S$ are invoked with zero density (i.e., if consider the long-run average of the 
frequency of invoking slave automata, then it is zero). This situation represents that monitoring with
slave automata happens with vanishing frequency which is a degenerate case.
(b)~The values returned by the slave automata from $S$ are bounded. It follows that these automata take transitions of non-zero weight only in some finite subword $w[i,j]$ of the input word $w$.
This situation represents monitoring of an infinite sequence, in which all events past position $j$ are irrelevant. This is a degenerate case in the infinite-word case. 
\end{remark}

The finite-width property does not depend on weights and hence we can construct an exponential-size \buchi{} automaton $\nonnestedA$, which simulates runs of 
a given NWA $\nestedA$. Having $\nonnestedA$, we can check whether it has a run corresponding to an accepting run of $\nestedA$, in which infinitely many backward-walking slave automata terminate 
at the same position. This check can be done in logarithmic space and hence checking the finite-width property is in $\PSPACE$. 
A simple reduction from the non-emptiness problem for NWA shows $\PSPACE$-hardness of checking the finite-width property.

\begin{restatable}{theorem}{FiniteWidthDecidable}
The problem whether a given NWA has finite width is $\PSPACE$-complete.
\label{th:FiniteWidthDecidable}
\end{restatable}

\subsection{The absolute sum value function} 
 
We present the main result on NWA of finite width.

\begin{restatable}{theorem}{ForwardAndBackward}
\label{th:FiniteWidth}
The emptiness problem for finite-width $(\flimavg;\fsum^+)$-automata with bidirectional slave automata is $\PSPACE$-hard and 
it is decidable in $\EXPSPACE$.
\end{restatable}	

\noindent\emph{Key ideas}. 
$\PSPACE$-hardness follows from $\PSPACE$-hardness of the emptiness problem for $(\flimavg;\fsum^+)$-automata with forward-walking slave automata only.
Containment in $\EXPSPACE$ is shown by reduction to the bounded-width case, which is shown decidable in the following section (Theorem~\ref{th:BoundedBF}).
We briefly describe this reduction. 
Consider a finite-width $(\flimavg;\fsum^+)$-automaton $\nestedA$ with bidirectional slave automata.
First, we observe that without loss of generality, we can assume that $\nestedA$ is deterministic.
Second, we observe that in every word accepted by $\nestedA$, at almost every position $i$ there exists 
a \emph{barrier}, which is a word $u$ such that 
(a)~the word $w' = w[1,i] u w[i+1,\infty]$, i.e., $w$ with $u$ inserted at position $i$, is accepted by $\nestedA$, and the runs on $w$ and $w'$ coincide except for positions in $w'$ corresponding to $u$,
(b)~in the run on $w'$, backward-walking slave automata active at the end of $u$  terminate within $u$,
(c)~in the run on $w'$, forward-walking slave automata active at the beginning of $u$ terminate within $u$,
and (d)~$u$ has exponential length. 
Basically, active slave automata cannot cross $u$ in $w'$ and in the effect insertion of $u$ bounds the number of active slave automata. 
Existence of barriers follows from the finite-width property of $\nestedA$.
 
We insert barriers in $w$ to reduce the number of active slave automata.
While inserting $u$ at a certain position may increase the partial average, 
we show that if at position $i$ in $w$, exponentially many active slave automata accumulates exponential weight past crossing $i$ (some slave automata walk forward while other backwards),
all partial averages  (of values returned by slave automata)  in $w'$ are bounded by the corresponding partial averages in $w$.
We conclude that for every word $w$, there exists a word $w'$ such that (i)~at every position at most exponentially many slave automata accumulate exponential values, and 
(ii)~the value of $w'$ does not exceed the value of $w$.
Thus, to compute the infimum over all runs of $\nestedA$, we can focus on runs in which at every position at most exponentially many slave automata accumulate exponential value.
Runs of slave automata in which they accumulate bounded (exponential) values can be eliminated as in Theorem~\ref{th:regularBF}, i.e., 
we can construct an exponential size NWA $\nestedA'$, which simulates $\nestedA$, and such that its slave automata run as long as they can accumulate value exponential (in $|\nestedA|$) and 
otherwise they non-deterministically pick the remaining value and the master automaton verifies that the pick is correct. 
Therefore, the infimum over all runs of $\nestedA$ coincides with the  infimum over all runs of $\nestedA'$ of width exponentially bounded.
\smallskip

\begin{remark}[Parametrized complexity]
If we assume that the size of slave automata in Theorem~\ref{th:FiniteWidth} is bounded by a constant, the complexity 
of the emptiness problem drops to $\NLOGSPACE$-complete. $\NLOGSPACE$-hardness follows from  $\NLOGSPACE$-hardness of the emptiness problem for
$\flimavg$-automata, which can be viewed as a special case of NWA. 
\end{remark}

\subsection{Expressive power}
\label{s:expressionPower}

\newcommand{\FWDBWD}{\mathcal{FB}(\flimavg; \fsum^+)}
\newcommand{\FWD}{\mathcal{F}(\flimavg; \fsum^+)}
\newcommand{\BWD}{\mathcal{B}(\flimavg; \fsum^+)}

DCP defined in Example~\ref{ex:data-consistency} can be expressed by a deterministic finite-width $(\flimavg;\fsum^+)$-automaton with bidirectional slave automata.
We show that both forward-walking and backward-walking slave automata are required to express DCP. 
That is, we formally show that DCP cannot be expressed by any (non-deterministic) $(\flimavg;\fsum^+)$-automaton with slave automata walking in one direction only.

\smallskip
\Paragraph{Classes of NWA}. 
We define $\FWDBWD$ as the class of all finite-width $(\flimavg;\fsum^+)$-automata with bidirectional slave automata.
We define $\FWD$ (resp., $\BWD$) as the subclass of $\FWDBWD$ consisting of NWA with forward-walking (resp., backward-walking) slave automata only.
 
We establish that classes $\FWD$ and $\BWD$ have incomparable expressive power and hence they are strictly less expressive than class $\FWDBWD$.  

\noindent\emph{Key ideas}. 
Consider a word  $w = (c \#^N  r^{2K} c \#^{2N} r^K )^{\omega}$ for some big $K$ and much bigger $N$.
An NWA from $\BWD$ computes DCP of $w$ by invoking (non-dummy) slave automata at every $r$ letter and taking silent transitions 
on letters $\#,c$. We show that an NWA $\nestedA$ from $\FWD$ cannot invoke the right number of slave automata, even if it uses non-determinism. 
More precisely, we show that $\nestedA$ computing DCP has to invoke at most $O(K)$ non-dummy slave automata on average on subwords $c \#^N  r^{2K} c \#^{2N} r^K$.
Since $N$ is much bigger than $K$, we conclude that $\nestedA$ has a cycle over $\#$ letters at which it takes only silent transitions and a cycle over $r$ letters on which it increases the multiplicity of active slave automata.
Using these two cycles, we construct a run of value smaller than DCP, which contradicts the assumption that $\nestedA$ computes DCP. 
Similarly, we can show that an NWA from $\BWD$ 
 cannot compute correctly DCP of words
of the form $w = (c  w^{2K} \#^N c w^K \#^{2N}) ^{\omega}$, while on these words DCP is expressible by an NWA from $\FWD$.

\begin{restatable}{lemma}{IncomparableForwardAndBackward}
(1)~DCP restricted to alphabet $\{r,\#,c\}$ is expressed by an NWA from $\BWD$, but it is not expressible by NWA from $\FWD$.
(2)~DCP restricted to alphabet $\{w,\#,c\}$ is expressed by an NWA from $\FWD$, but it is not expressible by NWA from $\BWD$.
\end{restatable}

The above lemma implies that DCP over alphabet $\{r,w,\#,c\}$ is not expressible by any NWA from $\FWD$ nor  from $\BWD$.   
In conclusion, we have:

\begin{theorem}
\label{th:expressivness}
(1)~$\FWD$ and $\BWD$ have incomparable expressive power.
(2)~$\FWDBWD$ are strictly more expressive than $\FWD$ and $\BWD$.
\end{theorem}

\section{Bounded-width case}
\label{s:bounded}
\newcommand{\cycle}{\mathcal{C}}
\newcommand{\AvgE}{\textsc{AvgE}}
\newcommand{\gain}{\textsc{Gain}}
\newcommand{\exNWA}{\nestedA}

In this section, we study 
$(\flimavg;\fsum)$-automata with bidirectional slave automata, which have bounded
width. The bounded width restriction
 has been introduced in~\cite{nested-mfcs} to improve the complexity of the emptiness problem and to establish decidability of 
the emptiness problem for $(\flimavg;\fsum)$-automata. 
NWA considered in~\cite{nested-mfcs} have only  forward-walking slave automata, while we extend these results to NWA with bidirectional slave automata. 
This extension preserves the complexity bounds from~\cite{nested-mfcs}, i.e., the emptiness problem is in $\PTIME$ for constant width and $\PSPACE$-complete for width given in unary.

The bounded width restriction emerges naturally in examples presented so far. If we bound the number of pending requests, we can express
ART and AW (Examples~\ref{ex:ART}~and~\ref{ex:dual-ART}) by automata of bounded width. 
If we bound the number of writes and reads between any two commits, then DCP and MR (Examples~\ref{ex:data-consistency}~and~\ref{ex:regret}) can be expressed by NWA of bounded width.
These natural restrictions lead to more efficient decision procedures. 

The decision procedure in this section differs from the one from~\cite{nested-mfcs}.
The key step in the decidability proof from~\cite{nested-mfcs} is establishing the following dichotomy: either the infimum over values of 
all words is $-\infty$ or the infimum is realized by \emph{dense} runs. A run is dense if for the values $v_1, v_2, \ldots $ returned by slave automata invoked at positions $1,2, \ldots$ we have $\frac{v_i}{i}$ converges to $0$, i.e., values returned by slave automata are sublinear in the positions of their invocation.
Properties of dense runs allow for further reductions, which lead to a decision procedure.
However, we show in the following Example~\ref{ex:runningBounded} that in case of NWA with bidirectional slave automata, dense runs may not attain the infimum of all runs.

\begin{example}
\label{ex:runningBounded}
Consider a $(\flimavg;\fsum)$-automaton $\exNWA$ with bidirectional slave automata over $\Sigma = \{a,b,c\}$. 
The NWA $\exNWA$ accepts words $(a b^* c)^{\omega}$ and it works as follows. 
On letters $a$, $\exNWA$ invokes a forward-walking slave automaton $\slaveA_a$, which returns the number of the following $b$ letters up to $c$.
On letters $c$, $\exNWA$ invokes a backward-walking slave automaton $\slaveA_c$, which returns the number of the preceding $b$ letters since $a$.
Finally, on $b$ letters, $\exNWA$ invokes a slave automaton $\slaveA_b$, which takes a single transition and returns value $0$.
The NWA $\exNWA$ has width $3$.
We can show that the value of any dense run, is $2$.
However, the infimum over values of all words is $1$. 
The partial average of the values returned by slave automata on finite word $u = (a b^* c)^*$ is $2$, while the partial average 
over $u a b^N$ is $\frac{2 |u| + N}{|u| + N}$. Therefore, the value, which is limit infimum over partial averages, of word 
$a b^{n_1} c \ldots a b^{n_i} c \ldots$ is $1$ if sequence $n_1, \ldots$ is grows rapidly (e.g. doubly-exponentially). 
\end{example}

\noindent\emph{Main ideas}.
In Example~\ref{ex:runningBounded}, the words attaining the least value contain long blocks of letter $b$, 
at which the NWA $\exNWA$ is (virtually) in the same state, i.e., it loops in this state. 
On letters $b$, the sum of all weights collected by all active slave automata is $2$, i.e., 
automata $\slaveA_a, \slaveA_c$ collect weight $1$ and $\slaveA_b$ collect $0$.
However, in computing limit infimum over partial averages, we pick positions just before letter $c$ as they correspond to the local minima, i.e.,
we compute the partial average over prefixes $u a b^N$, and hence
the weights collected by $\slaveA_c$ do not contribute to this partial average.
Then, the sum of all weights collected by slave automata $\slaveA_a, \slaveA_b$ over a letter $b$ is $1$, which is 
equal to the least value of the limit infimum of the partial averages. 
In the following, we extend this idea and present the solution of all 
bounded-width $(\flimavg;\fsum)$-automata with bidirectional slave automata.
We show that the infimum over all words of a given NWA is the least average value over all cycles. 
In the following, we define appropriate notions of cycles of NWA and their average with exclusion of some slave automata.

\Paragraph{The graph of $k$-configurations}.
Let $\nestedA$ be a non-deterministic $(\flimavg; \fsum)$-automaton of width $k$.
We define a \emph{$k$-configuration} of $\nestedA$ as a tuple 
$(q; q_1, \ldots, q_k)$ where $q$ is a state of the master automaton, and
each $q_1, \ldots, q_k$ is either a state of a slave automaton of $\nestedA$ 
or $\bot$.
Given a run of $\nestedA$, we say that $(q; q_1, \ldots, q_k)$  is the $k$-configuration at position $i$ in the run if
$q$ is the state of the master automaton at position $i$ and there are $l \leq k$ active slave automata at position $i$, whose states are $q_1, \ldots, q_l$ 
ordered by position of invocation (backward-walking slave automata are invoked past position $i$).
If $l < k$, then $q_{l+1}, \ldots, q_{k} = \bot$. 
We say that a $k$-configuration $C_2$ is a successor of a $k$-configuration $C_1$ if there exists an accepting run of $\nestedA$ and $i>0$
such that $C_1$ is the $k$-configuration at $i$ and $C_2$ is the $k$-configuration at $i+1$ 	
The \emph{graph of $k$-configurations} of $\nestedA$ is the set of $k$-configurations  of $\nestedA$, which occur infinitely often in some accepting run,  with the successor relation.

\begin{figure}
\centering
\begin{tikzpicture}
\newcommand{\drawRun}[3]{
\draw (#1*0.3,0.15) -- (#1*0.3,#3);
\draw[->] (#1*0.3,#3) -- (#2*0.3,#3); 

}
\drawRun{0}{9}{1.2}
\drawRun{1}{7	}{1}
\drawRun{2}{10}{0.8}
\drawRun{4}{9}{0.4}
\drawRun{11}{1}{1.4}

\node at (3*0.3,-0.4) {$i$};	

\node at (7*0.3,-0.4) {$C$};	

\draw[densely dashed] (5*0.3+0.15,-0.5) -- (5*0.3+0.15,0.5) ; 
\draw[densely dashed] (8*0.3+0.15,-0.5) -- (8*0.3+0.15,0.5) ; 

\draw[ultra thick] (5*0.3+0.15,1.2) -- (8*0.3+0.15,1.2);
\draw[ultra thick] (5*0.3+0.15,0.8) -- (8*0.3+0.15,0.8);

\foreach \i in {0,...,11}
{
\node[rectangle,draw,minimum width=0.3cm,minimum height=0.3cm] at (0.3*\i,0) {};
}

\begin{scope}[xshift=4cm]

\foreach \i in {1,...,11}
{
\node[rectangle,draw,minimum width=0.3cm,minimum height=0.3cm] at (0.3*\i,0) {};
}

\node at (9*0.3,-0.4) {$j$};	

\node at (6*0.3,-0.4) {$C$};	

\draw[densely dashed] (3*0.3+0.15,-0.5) -- (3*0.3+0.15,0.5) ; 
\draw[densely dashed] (7*0.3+0.15,-0.5) -- (7*0.3+0.15,0.5) ; 

\draw[ultra thick] (3*0.3+0.15,1) -- (7*0.3+0.15,1);
\draw[ultra thick] (3*0.3+0.15,0.8) -- (7*0.3+0.15,0.8);
\draw[ultra thick] (3*0.3+0.15,0.4) -- (5*0.3+0.23,0.4);

\draw[ultra thick] (3*0.3+0.15,0.6) -- (4.6*0.3,0.6);
\draw[ultra thick] (5*0.3,0.6) -- (7*0.3+0.15,0.6);

\drawRun{4}{6}{0.4}
\drawRun{5}{10}{0.6}
\drawRun{2}{4.8}{0.6}
\drawRun{1}{9}{0.8}

\drawRun{8}{1}{1}
\drawRun{9}{2}{1.2}
\drawRun{10}{3}{1.4}

\end{scope}

\end{tikzpicture}
\caption{Pictorial explanation of  $\gain(\cycle,FC)$ (on the left) and $\AvgE(\cycle, R)$ (on the right). The gain $\gain(\cycle,FC)$ on the left is the sum of weights corresponding to thick parts of runs of slave automata invoked before $i$.
The average $\AvgE(\cycle, R)$ corresponds to the average of the thick parts of runs divided by the number of slave automata invoked within $C$. Slave automata invoked past $j$ are excluded from the average.
}
\label{fig:focus}
\end{figure}
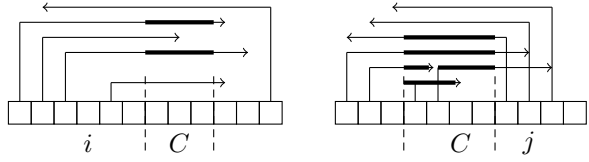

\Paragraph{Characteristics of cycles}.
Let $\cycle$ be a cycle in a graph of $k$-configurations of $\nestedA$.
Let $F$ (resp., $B$) be the set of forward-walking (resp., backward-walking) slave automata, which are active throughout $\cycle$, i.e., which are  not invoked nor terminated within $\cycle$.
A \emph{focus} $Fc$ (for $\cycle$) is a downward closed subset of $F$, i.e., it contains all automata from $F$ invoked before some position.	
We define a focused gain $\gain(\cycle,Fc)$ as the sum of weights which automata from $Fc$ accumulate over $\cycle$.
A \emph{restriction} $R$ (for $\cycle$) is an upward closed subset of $B$, i.e., it contains all automata from $B$ invoked past certain position.	
We define an average weight of $\cycle$ excluding $R$, denoted by $\AvgE(\cycle,R)$, as the sum of weights of all transitions of slave automata within $\cycle$, 
except of transitions of slave automata from $R$, divided by the number of slave automata invoked within $\cycle$.

Intuitively, a focused gain refers to the value, which forward-walking slave automata invoked before some position $i$, accumulate over 
the part of run corresponding to $\cycle$ (see Figure~\ref{fig:focus}).
If the focused gain $\gain(\cycle,Fc)$ is negative, then by pumping $\cycle$ we can arbitrarily decrease the partial average of the values of slave automata invoked before $i$.
In consequence, we can construct a run of the value $-\infty$.
Formally, we define condition (*), which implies that there exists a run of value $-\infty$, as follows: (*)~there exists a cycle $\cycle$ in the graph of $k$-configurations of $\nestedA$ and 
a focus $Fc$ such that $\gain(\cycle,Fc) < 0$. 

If the focused gain of every cycle is non-negative, we need to examine averages of cycles, while excluding some backward-walking slave automata.
The average weight with restriction corresponds to the partial average of values aggregated over $\cycle$ by all slave automata invoked before position $j$ (which can be past $\cycle$). 
Backward-walking slave automata in the restriction correspond to automata invoked past $j$, and 
hence their values do not contribute to the partial average (up to $i$) (see Figure~\ref{fig:focus}).
In Example~\ref{ex:runningBounded}, we compute the average of slave automata over letters $b$, but we exclude the backward-walking slave automaton invoked at the following letter $c$.
Observe that for any cycle $\cycle$ and any restriction $R$, having a run containing $\cycle$ occurring infinitely often, 
we can repeat each occurrence of cycle $\cycle$ sufficiently many times so that the partial average of values of slave automata up to position corresponding to $j$ becomes arbitrarily close to the average $\AvgE(\cycle, R)$.
The resulting run contains a subsequence of partial averages convergent to $\AvgE(\cycle, R)$
and hence its value does not exceed $\AvgE(\cycle, R)$.
We can now state our key technical lemma. 
This lemma is a direct extension of an intuition behind computing the infimum over values of all words of the NWA $\exNWA$ from Example~\ref{ex:runningBounded}.

\begin{restatable}{lemma}{TechnicalBoundedWidth}
\label{l:techlical-bounded}
Let $\nestedA$ be a $(\flimavg;\fsum)$-automaton of bounded width with bidirectional slave automata. 
(1)~If condition (*) holds, then 
$\nestedA$ has a run of value $-\infty$. 
(2)~If (*) does not hold, then 
the infimum $\inf_w \nestedA(w)$ equals
the infimum $\inf_{\cycle \in \Lambda, R} \AvgE(\cycle, R)$, where 
$\Lambda$ is the set of all cycles $\cycle$ in the graph of $k$-configurations of $\nestedA$.
\end{restatable}

If the width of $\nestedA$ is constant, then the graph of $k$-configurations has polynomial size in $|\nestedA|$ and it can be constructed in polynomial time by employing
reachability checks on the set of all $k$-configurations w.r.t. to relaxation of the successor relation.
Therefore, for every focus $Fc$ and every $k$-configuration $c$ we can check in polynomial time whether there exists a cycle $\cycle$ such that
$\cycle[1] = c$ and $\gain(\cycle, Fc) < 0$. 
Thus, condition (1) can be check in logarithmic space assuming that weights are given in unary. If weights are given in binary, condition (1) can be checked in polynomial time.
Checking condition (2) has the same complexity as condition (1).
If the width $k$ is given in unary in input,  the graph of $k$-configurations is exponential in $|\nestedA|$ and
conditions (1) and (2) can be checked in polynomial space. 
Weights in this case are logarithmic in the size of the graph and hence changing representation from binary to unary does not affect the (asymptotic) size of the graph.

\begin{restatable}{theorem}{BoundedWidthForwardAndBackward}
\label{th:BoundedBF}
The emptiness problem for $(\flimavg;\fsum)$-automaton of width $k$ with bidirectional slave automata is 
(a)~$\NLOGSPACE$-complete for constant $k$ and weights given in unary, 
(b)~in $\PTIME$ for constant $k$ and weights given in binary, and
(c)~$\PSPACE$-complete for $k$ given in unary.
\end{restatable}

\section{Extensions}
\label{s:extensions}
In this section we briefly discuss some extensions of the model of bidirectional NWA, i.e.,
we discuss the possibility of invoking multiple slave automata in one 
transition and two-way walking slave automata.

\begin{figure}
\centering
\begin{tikzpicture}
\newcommand{\drawLoop}[4]
{
\draw[rounded corners] (#1,#2) -- (#1+#3,#2) -- (#1+#3,#2-#4) -- (#1,#2-#4);		
}

\drawLoop{0}{0}{1}{0.3}
\drawLoop{0}{-0.3}{-1.5}{0.3}
\drawLoop{0}{-0.6}{1.3}{0.3}

\draw[->] (0,-0.9) -- (-1,-0.9	);

\node[draw,circle, fill,minimum size=0.1cm,inner sep=0] at (0,0) {};
\node[draw,circle, fill,minimum size=0.1cm,inner sep=0] at (0,-0.3) {};
\node[draw,circle, fill,minimum size=0.1cm,inner sep=0] at (0,-0.6) {};
\node[draw,circle, fill,minimum size=0.1cm,inner sep=0] at (0,-0.9) {};

\node at (1.5,-0.2) {$\slaveA_1$};
\node at (1.7,-0.8) {$\slaveA_3$};

\node at (-1.9,-0.45) {$\slaveA_2$};
\node at (-1.7,-0.9) {$\slaveA_4$};

\draw[densely dashed] (0,-1.3) -- (0,0.4);

\foreach \i in {1,...,18}
{
\node[rectangle,draw,minimum width=0.3cm,minimum height=0.3cm] at (0.3*\i-3,-1.5) {};
}

\end{tikzpicture}
\caption{A run of a two-way slave automaton $\slaveA$ and 
slave automata $\slaveA_1, \ldots, \slaveA_4$ simulating loops of $\slaveA$.}
\label{fig:twoway}
\end{figure}
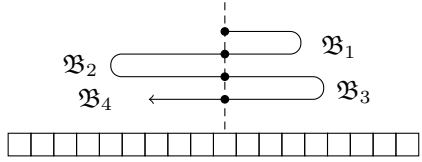

\Paragraph{Invocation of multiple slave automata}.
The master automaton of an NWA invokes exactly one slave automaton at every transition. 
We can generalize the definition of NWA and allow the master automaton to invoke up to some 
constant $k$ slave automata at every transition. 
We call such a model $k$-NWA. 
First note that $k$-NWA contain NWA. 
Conversely, 
we briefly describe a reduction of the emptiness problem for $k$-NWA to the emptiness problem for  NWA.   
First, observe that without loss of generality we can assume that $k$-NWA always invokes exactly $k$-slave automata as it can 
always invoke a dummy slave automaton, which immediately accepts. 
Invocation of such a slave automaton is equivalent to taking a silent transition. 
Next, given a $k$-NWA $\nestedA$ with bidirectional slave automata over the alphabet $\Sigma$, we can construct an NWA $\nestedA'$ with bidirectional slave automata
over the alphabet $\Sigma \cup \{ \# \}$, which accepts words of the form $w[1] \#^{k-1} w[2] \#^{k-1} \ldots$. 
The runs of $\nestedA'$ on the word $w[1] \#^{k-1} w[2] \#^{k-1} w[3] \ldots$ correspond to all runs of $\nestedA$ on $w$;
the $\nestedA'$ invokes at letters $w[i] \#^{k-1}$ exactly $k$ slave automata which $\nestedA$ invokes at the corresponding transition over letter $w[i]$. 
Observe that the emptiness problems for $\nestedA$ and $\nestedA'$ coincide.
\smallskip

\Paragraph{Two-way walking slave automata}. 
For the ease of presentation we focus on bidirectional NWA where each slave automaton is either
forward walking or backward walking. 
However, in general, we can allow slave automata that change direction while running, i.e., 
allow two-way slave automata and obtain the same complexity results.
Observe that in case of two-way $\fsum^+$-automata, we can assume that such an automaton does not visit the same position in the same state twice. 
Indeed, such a cycle can be eliminated without increase of the value of the run. 
Thus, without loss of generality we assume that every two-way slave automaton visits every position at most $|\slaveA|$ times.
Therefore, instead of invoking a two-way slave $\fsum^+$-automaton the master automaton invokes  multiple forward-walking and backward walking slave automata,
two automata, a forward
walking $\slaveA_f$ and a backward walking $\slaveA_b$ such that $\slaveA_f$ (resp., $\slaveA_b$) 
simulates the run of $\slaveA$ past its invocation position (resp., before its invocation position).
This reduction shows that every $(\flimavg;\fsum^+)$-automaton $\nestedA$ with two-way walking slave automata is equivalent 
to a $2$-NWA with bidirectional slave automata. 
This reduction however involves exponential blow-up as slave automata $\slaveA_f, \slaveA_b$ can have exponential size in $\nestedA$. 
This follows from the fact that due to reversals each visited position by $\slaveA$ can be visited $|\slaveA|$ times in different states. 
To simulate that in one run, $\slaveA_f$ and $\slaveA_b$ have to simulate $|\slaveA|$ instances of $\slaveA$ in different states. 
This exponential blow-up can be avoided by dividing $\slaveA_f, \slaveA_b$ into multiple slave automata, 
each of which tracks only one loop in the run of $\slaveA$ as shown in Figure~\ref{fig:twoway} with automata $\slaveA_1, \ldots, \slaveA_4$. 
Each of these automata have to track at most two instances of $\slaveA$ and hence it involves only quadratic blow-up. 
The resulting automaton is an $|\nestedA|$-NWA as up to $\frac{|\slaveA|}{2}$ slave 
automata have to be invoked at every position. 
Still, the emptiness problems for $|\nestedA|$-NWA and for NWA have the same complexity and hence we conclude that 
allowing two-way slave automata does not increase the emptiness problem for $(\flimavg;\fsum^+)$-automata.

\section{Discussion and Conclusion}

\smallskip\noindent{\em Discussion.}
We established decidability of the emptiness problem for classes of bidirectional NWA, which include all 
NWA presented in the examples.
An NWA from Example~\ref{ex:energy} is covered by Theorem~\ref{th:regularBF},
while NWA from Examples~\ref{ex:dual-ART}, \ref{ex:data-consistency} and~\ref{ex:regret} are covered by 
Theorem~\ref{th:FiniteWidth}.
The lower bounds in the presented theorems follow from the lower bounds of the special case of forward-only NWA.
The established complexity bounds (Table~\ref{tab:complexity}) coincide with the bounds 
for forward-only NWA.

\begin{table} 
\centering
\begin{tabular}{|c|c|c|c|}
\hline
Value   & \multirow{2}{*}{Restrictions} & Complexity & Complexity \\
func. $g$ & & 
 Bidirectional & Forward case \\ 
\hline
$\fmin$,$\fmax$, & \multirow{2}{*}{None} & \textbf{$\PSPACE$-complete} & \multirow{2}{*}{$\PSPACE$-complete~\cite{nested}}\\
$\fBsum{} $ & & (Thm~\ref{th:regularBF}) & \\
\hline
 & \multirow{2}{*}{finite} &  \textbf{$\PSPACE$-hard} & \multirow{2}{*}{$\PSPACE$-hard} \\
$\fsum^+$ & \multirow{2}{*}{width} & \textbf{$\EXPSPACE$}  & \multirow{2}{*}{$\EXPSPACE$~\cite{nested}} \\
 & & (Thm~\ref{th:FiniteWidth}) & \\
\hline
$\fsum^+$, & constant width & \textbf{$\NLOGSPACE$-complete} & \multirow{1}{*}{$\NLOGSPACE$-complete} \\
$\fsum$&  unary weights & (Thm~\ref{th:BoundedBF}) & \cite{nested-mfcs} \\
\hline
$\fsum^+$, & constant  width& \textbf{$\PTIME$} & \multirow{2}{*}{$\PTIME$~\cite{nested-mfcs}} \\
$\fsum$&  binary weights & (Thm~\ref{th:BoundedBF}) & \\
\hline
$\fsum^+$, & width given & \textbf{$\PSPACE$-complete} & \multirow{2}{*}{$\PSPACE$-complete~\cite{nested-mfcs}} \\
$\fsum$&  in unary & (Thm~\ref{th:BoundedBF}) & \\
\hline
\end{tabular} 
\caption{The complexity of the emptiness problem for $(\flimavg;g)$-automata. The columns describe respectively: the value function $g$, 
additional restrictions imposed on the problem, the complexity in the case with bidirectional slave automata, and
the complexity in the previously studied~\cite{nested,nested-mfcs} case with only forward-walking slave automata. 
Results presented in this paper are boldfaced. }
\label{tab:complexity}
\end{table}

\smallskip\noindent{\em Concluding remarks.}
In this work we present bidirectional NWA as a specification formalism for quantitative 
properties. In this formalism  many natural quantitative properties can be expressed, and 
we present decidability and complexity results for the basic decision problems.
There are several interesting directions for future work.
The study of bidirectional NWA with other value functions is an interesting direction.
The second direction of future work is to consider other formalism (such as a logical 
framework) which has the same expressive power as bidirectional NWA.

{\footnotesize
\subparagraph*{Acknowledgements.}
This research was supported in part by the Austrian Science Fund (FWF) under grants S11402-N23,S11407-N23 (RiSE/SHiNE) and Z211-N23 (Wittgenstein Award), 
ERC Start grant (279307: Graph Games), Vienna Science and Technology Fund (WWTF) through project ICT15-003 and
by the National Science Centre (NCN), Poland under grant 2014/15/D/ST6/04543.
}

\bibliography{papers}

\begin{thebibliography}{10}

\bibitem{AlmagorBK14}
Shaull Almagor, Udi Boker, and Orna Kupferman.
\newblock Discounting in {LTL}.
\newblock In {\em {TACAS}, 2014}, pages 424--439, 2014.

\bibitem{DBLP:journals/jacm/AlmagorBK16}
Shaull Almagor, Udi Boker, and Orna Kupferman.
\newblock Formally reasoning about quality.
\newblock {\em J. {ACM}}, 63(3):24:1--24:56, 2016.

\bibitem{DBLP:conf/lics/AlurDDRY13}
Rajeev Alur, Loris D'Antoni, Jyotirmoy~V. Deshmukh, Mukund Raghothaman, and
  Yifei Yuan.
\newblock Regular functions and cost register automata.
\newblock In {\em LICS 2013}, pages 13--22, 2013.

\bibitem{boundsInWRegularity}
Miko{\l}aj Boja{\'{n}}czyk and Thomas Colcombet.
\newblock Bounds in w-regularity.
\newblock In {\em 21th {IEEE} Symposium on Logic in Computer Science {(LICS}
  2006), 12-15 August 2006, Seattle, WA, USA, Proceedings}, pages 285--296,
  2006.

\bibitem{BokerCHK14}
Udi Boker, Krishnendu Chatterjee, Thomas~A. Henzinger, and Orna Kupferman.
\newblock Temporal specifications with accumulative values.
\newblock {\em {ACM} {TOCL}}, 15(4):27:1--27:25, 2014.

\bibitem{bollig2010pebble}
Benedikt Bollig, Paul Gastin, Benjamin Monmege, and Marc Zeitoun.
\newblock Pebble weighted automata and transitive closure logics.
\newblock In {\em {ICALP} 2010, Part {II}}, pages 587--598. Springer, 2010.

\bibitem{BouyerMM14}
Patricia Bouyer, Nicolas Markey, and Raj~Mohan Matteplackel.
\newblock Averaging in {LTL}.
\newblock In {\em {CONCUR} 2014}, pages 266--280, 2014.

\bibitem{DBLP:journals/corr/BouyerMRLL15}
Patricia Bouyer, Nicolas Markey, Mickael Randour, Kim~Guldstrand Larsen, and
  Simon Laursen.
\newblock Average-energy games.
\newblock In {\em GandALF 2015.}, pages 1--15, 2015.

\bibitem{DBLP:journals/corr/abs-1006-1492}
Krishnendu Chatterjee, Laurent Doyen, Herbert Edelsbrunner, Thomas~A.
  Henzinger, and Philippe Rannou.
\newblock Mean-payoff automaton expressions.
\newblock In {\em {CONCUR}}, pages 269--283, 2010.

\bibitem{Chatterjee:2009:AWA:1789494.1789497}
Krishnendu Chatterjee, Laurent Doyen, and Thomas~A. Henzinger.
\newblock Alternating weighted automata.
\newblock In {\em FCT'09}, pages 3--13. Springer, 2009.

\bibitem{DBLP:journals/corr/abs-1007-4018}
Krishnendu Chatterjee, Laurent Doyen, and Thomas~A. Henzinger.
\newblock Expressiveness and closure properties for quantitative languages.
\newblock {\em LMCS}, 6(3), 2010.

\bibitem{Chatterjee08quantitativelanguages}
Krishnendu Chatterjee, Laurent Doyen, and Thomas~A. Henzinger.
\newblock Quantitative languages.
\newblock {\em ACM TOCL}, 11(4):23, 2010.

\bibitem{nested}
Krishnendu Chatterjee, Thomas~A. Henzinger, and Jan Otop.
\newblock Nested weighted automata.
\newblock In {\em {LICS} 2015}, pages 725--737, 2015.

\bibitem{nested-mfcs}
Krishnendu Chatterjee, Thomas~A. Henzinger, and Jan Otop.
\newblock Nested weighted limit-average automata of bounded width.
\newblock In {\em {MFCS} 2016}, pages 24:1--24:14, 2016.

\bibitem{nestedprob}
Krishnendu Chatterjee, Thomas~A. Henzinger, and Jan Otop.
\newblock Quantitative automata under probabilistic semantics.
\newblock In {\em {LICS} 2016}, pages 76--85, 2016.

\bibitem{nested-sas}
Krishnendu Chatterjee, Thomas~A. Henzinger, and Jan Otop.
\newblock Quantitative monitor automata.
\newblock In {\em {SAS} 2016}, pages 23--38, 2016.

\bibitem{DBLP:journals/tac/ChatterjeeP15}
Krishnendu Chatterjee and Vinayak~S. Prabhu.
\newblock Quantitative temporal simulation and refinement distances for timed
  systems.
\newblock {\em {IEEE} Trans. Automat. Contr.}, 60(9):2291--2306, 2015.

\bibitem{Droste:2009:HWA:1667106}
Manfred Droste, Werner Kuich, and Heiko Vogler.
\newblock {\em Handbook of Weighted Automata}.
\newblock Springer, 1st edition, 2009.

\bibitem{DrosteR06}
Manfred Droste and George Rahonis.
\newblock Weighted automata and weighted logics on infinite words.
\newblock In {\em {DLT} 2006}, pages 49--58, 2006.

\bibitem{copylessCRA}
Filip Mazowiecki and Cristian Riveros.
\newblock Copyless cost-register automata: Structure, expressiveness, and
  closure properties.
\newblock In {\em {STACS}, 2016}, pages 53:1--53:13, 2016.

\bibitem{DBLP:journals/jalc/Mohri02}
Mehryar Mohri.
\newblock Semiring frameworks and algorithms for shortest-distance problems.
\newblock {\em J. Aut. Lang. \& Comb.}, 7(3):321--350, 2002.

\end{thebibliography}

\newpage
\appendix
\section*{Appendix}

In the appendix we recall statements of theorems and lemmas from the main body of the paper keeping their original numbering. 
Lemmas introduced in the appendix have subsequent numbers. For this reason, the numbering  of theorem and lemmas in 
the appendix is mixed. 

\section{Proofs from Section~\ref{s:decision}}
\RegularForwardAndBackward*

The $\PSPACE$-hardness part follows from $\PSPACE$-hardness of the emptiness problem for 
 $(\flimavg;g)$-automata with forward-walking slave automata only~\cite{nested}. Therefore, we focus on 
the containment in $\PSPACE$. We begin with a definition of a unifying framework of regular value functions. 

\Paragraph{Regular weighted automata and regular value functions}. 
Following~\cite{nested}, we say that a weighted automaton $\aut$ over finite words is a \emph{regular weighted automaton} 
if and only if
there is a finite set of rationals $\{ q_1, \ldots, q_n\} $ and 
there are regular languages ${\lang}_1, \ldots, {\lang}_n$
such that 
\begin{compactenum}[(i)]
\item every word accepted by $\aut$ belongs to $\bigcup_{1 \leq i \leq n} {\lang}_i$, and
\item for every $w \in {\lang}_{i}$, its value w.r.t. $\aut$ is $q_i$.
\end{compactenum}
A value function $f$ is a \emph{regular value function}
if and only if all $f$-automata are regular weighted automata.
Examples of regular value functions are $\fmin,\fmax$ and variants of the bounded sum $\fBsum{B}$ with regular conditions, i.e., the partial sum of every prefix, suffix, infix of the run belongs to interval $L,U$, e.t.c.
Having the definition of regular value function, we can easily check whether our variant of the bounded sum is admissible.

We define the \emph{description size} of a given regular weighted automaton $\aut$,
as the size of automata $\aut_1, \ldots, \aut_n$ recognizing languages 
 ${\lang}_1, \ldots, {\lang}_n$ that witness $\aut$ being a regular weighted automaton.

\begin{lemma}
Let $g$ be a regular value function. 
Every $(\flimavg;g)$-automaton $\nestedA$ with bidirectional slave automata is equivalent to an exponential-size 
$\flimavg$-automaton $\nonnestedA$. The automaton $\nonnestedA$ can be constructed implicitly in polynomial time.
\label{l:regularEquivalent}
\end{lemma}

Since the emptiness problem for $\flimavg$-automata can be solved in $\NLOGSPACE$, 
Lemma~\ref{l:regularEquivalent} implies Theorem~\ref{th:regularBF}. Moreover, Remark~\ref{rem:parametricRegular}
follows directly from the construction in the following proof.

\begin{proof}
\newcommand{\QF}{Q_{SF}}
\newcommand{\QB}{Q_{SB}}
Let $Q_m$ be the set of states of the master automaton of $\nestedA$.
Since $g$ is a regular value function, every slave automaton $\slaveA$ has a finite range $R_{\slaveA}$ and for each value $v \in R_{\slaveA}$ there exists a finite-state automaton 
$\aut_{\slaveA,v}$ recognizing the set of words of value $v$. 
Let $\QF$ (resp., $\QB$) be the union of the sets of states of all automata $\aut_{\slaveA,v}$, where $\slaveA$ if a forward-walking (resp., backward-walking) slave automaton and $v \in R_{\slaveA}$.
We assume that sets of states of automata $\aut_{\slaveA,v}$ are disjoint.	

We define a $\flimavg$-automaton $\nonnestedA$ with generalized \buchi{} condition as follows. 
The set of states of $\nonnestedA$ is $Q_m \times 2^{\QF} \times 2^{\QF} \times 2^{\QB}$.
States of $\nonnestedA$ are of the form $(q,F_1, F_2,  B)$, whose objectives are as follows.

Component $q$ is used to simulate the run of the master automaton. 
Components $F_1, F_2$ are used to simulate finite-state automata corresponding to forward-walking slave automata.
Accepting states correspond to termination of a slave automaton and hence they are removed from $F_1, F_2$.
Every newly invoked forward-walking slave automaton is added to component $F_2$, i.e., automaton $\nonnestedA$ picks a transition invoking slave automaton $\slaveA$ and picks weight of this transition $v$,
then it adds to $F_2$ an initial state of finite-state automaton $\aut_{\slaveA, v}$.  The weight of such a transition is $v$,

If every slave automaton has finite run than $F_1$ becomes empty at some point of time.
Then, we move all states from $F_2$ to $F_1$, and put $F_2 = \emptyset$. 
Observe that runs of all forward-walking slave automaton are finite if and only if $F_1$ is empty infinitely often.

Component $B$ is used to simulate finite-state automata corresponding to  backward-walking slave automata.
Accepting states of backward-walking slave automata correspond to their termination. 
Since $\nonnestedA$ moves in the opposite direction w.r.t. $\nestedA$, the automaton $\nonnestedA$ 
adds to component $B$ some subset of accepting states from $\QB$.
There is only one position in the run of $\nonnestedA$ at which it adds states to $B$.
Whenever a backward-walking slave automaton is invoked in a state $q_i$, we require that
$q_i$ belongs to $B$. This state $q_i$ may be removed from the corresponding set but does not have to. 
Removal corresponds to a situation when only a single slave automaton is in the state $q_i$, while leaving $q_i$ in $B$
corresponds the situation when more than one slave automaton is in state $q_i$.

Due to the construction we have (i)~for every run of $\nestedA$, the automaton $\nonnestedA$ has a run of the same weight, and conversely
(ii)~for every run of $\nonnestedA$, there exists a run $\nestedA$ of the same value. 
\end{proof}
 
\section{Proofs from Section~\ref{s:finite}}

	\newcommand{\Qslv}{Q_{\textrm{slv}}}	
\newcommand{\confNum}[1]{\mathsf{conf}({#1})}
\newcommand{\mult}{{\mathsf{mult}}}
\newcommand{\boundOnMulti}{\mathbf{N}}
\newcommand{\cond}{\mathbb{C}}

\subsection{The proof of Theorem~\ref{th:FiniteWidthDecidable}}

\FiniteWidthDecidable*
\begin{proof}
\newcommand{\QslvF}{Q_{slv}^F}
\newcommand{\QslvB}{Q_{slv}^B}
Let $\QslvF$ (resp., $\QslvB$) be the set of states of forward-walking (resp., backward-walking) slave automata of $\nestedA$ and let $Q_m$ be the set of states of the master automaton of $\nestedA$.
We define a (generalized) \buchi{}-automaton (with no weights) $\nonnestedA$ as follows. 
The set of states of $\nonnestedA$ is $Q_m \times 2^{\QslvF} \times 2^{\QslvF} \times 2^{\QslvB} \times 2^{\QslvB}$.
Initially $\nonnestedA$ starts with $(q,\emptyset,\emptyset,\emptyset,\emptyset)$, where $q$ is some initial state of the master automaton.
We use sets of states to simulate runs of slave automata and to ensure that every forward-walking slave automaton runs for finitely many steps. 
We treat backward-walking slave automata in a similar way to forward-walking slave automata except that backward-walking slave automata are started at the termination step
of the corresponding slave automata, and they can terminate (which correspond to invocation) multiple, but finitely many times. 
More precisely, states of $\nonnestedA$ are of the form $(q,F_1, F_2, B_1, B_2)$, whose objectives are as follows.

Component $q$ is used to simulate the run of the master automaton. Components $F_1, F_2$ are used to simulate forward-walking slave automata.
Accepting states correspond to termination of a slave automaton and hence they are removed from $F_1, F_2$.
A newly invoked forward-walking slave automaton is added to component $F_2$. Thus, if every slave automaton has finite run than $F_1$ becomes empty at some point of time.
Then, we move all states from $F_2$ to $F_1$, and put $F_2 = \emptyset$. Observe that runs of all forward-walking slave automaton are finite if and only if $F_1$ is empty infinitely often.

Components $B_1, B_2$ are used to simulate backward-walking slave automata.
Component $B_1$ contains the states of slave automata that all finish at the same position, while $B_2$ contains states of other slave automata. 
We require that $B_1$ and $B_2$ are disjoint at every position.
Accepting states of backward-walking slave automata correspond to their termination. 
However, as the simulating automaton $\nonnestedA$ moves in the opposite direction, it guesses 
at every step whether some (backward-walking) slave automaton
have been terminated at the current position and it may add some accepting states to $B_1$ or $B_2$.
There is only one position in the run of $\nonnestedA$ at which it adds states to $B_1$.
Whenever a backward-walking slave automaton is invoked in a state $q_i$, we require that
$q_i$ belongs to $B_1$ or $B_2$. This state $q_i$ may be removed from the corresponding set but does not have to. 
Removal corresponds to a situation when only a single slave automaton is in the state $q_i$, while leaving $q_i$ in $B_1$ or $B_2$
corresponds the situation when more than one slave automaton is in state $q_i$.
 
Consider an infinite run $\pi$ of $\nonnestedA$, in which
(a)~component $q$ equals infinitely often to some accepting state of the master automaton,
(b)~component $F_1$ is empty infinitely often, and
(c)~component $B_1$ is never empty and infinitely often contains an initial state of the slave automaton invoked at the current position.
The run $\pi$ corresponds to an accepting run of $\nestedA$ in which infinitely many backward-walking slave automata terminate at the same position (which is the position at which $B_1$ becomes non-empty for the first time).
Conversely, having a run of $\nestedA$ of infinite width, the corresponding run of $\nonnestedA$ satisfies (a), (b) and (c).
Checking existence of a run of $\nonnestedA$ satisfying  (a), (b) and (c) can be done in non-deterministic logarithmic space. 
Thus, checking whether $\nestedA$ has infinite width can be done in polynomial space. 

We show $\PSPACE$-hardness of checking finite-width, by reduction from the emptiness problem for NWA with forward-walking slave automata only, which is $\PSPACE$-complete~\cite{nested}.
Let $\nestedA$ be an NWA with forward-walking slave automata only over the alphabet $\Sigma$. We extend the alphabet $\Sigma$ by $\$,\#$.
Next, we construct an NWA $\nestedA'$ whose every run has infinite width and which accepts precisely words of the form $\$ w[1] \# w[2] \# $ such that $\nestedA$ accepts $w[1] w[2] \ldots$. 
Basically, forward-walking slave automata of $\nestedA$ ignore letters $\#$ and  every transition $(q,a,q')$ of the master automaton of $\nestedA$ is replaced with two transitions 
$(q,a,q'_{\#})$ and $(q'_{\#}, \#,q')$. On the first transition, $\nestedA'$ invokes the same automaton as $\nestedA$ and on transitions labeled by $\#$, the automaton $\nestedA'$ invokes a backward-walking slave automaton
that runs until $\$$. Now, if $\nestedA$ does not have an accepting run, $\nestedA'$ does not have an accepting run and it has trivially finite width.
Otherwise, if $\nestedA$ accepts $w$, the automaton $\nestedA'$ accepts $\$ w[1] \# w[2] \# $ and it has a run of infinite width on it.
Therefore, $\nestedA$  has an accepting run if and only if $\nestedA'$ does not have finite width.

\end{proof}

\subsection{The proof of Theorem~\ref{th:FiniteWidth}}

\ForwardAndBackward*

The $\PSPACE$-hardness part follows from $\PSPACE$-hardness of the emptiness problem for 
 $(\flimavg;\fsum^+)$-automata with forward-walking slave automata only~\cite{nested}. Therefore, we focus on 
the containment in $\EXPSPACE$. 

\Paragraph{Overview}.
The proof is by reduction to the bounded-width case, which is decidable due to Theorem~\ref{th:BoundedBF}. 
First, we show that without loss of generality we can assume that a given NWA $\nestedA$ is  deterministic (Lemma~\ref{WLOGDeterministicLimAvgSum}).
Next, we define words, called barriers, upon which all active (backward- and forward-walking) slave automata terminate.
We show that for finite-width NWA such words do exist (Lemma~\ref{l:barriersExist}). 
Properties of barriers ensure that if as some position $i$ in the input word, exponentially many slave automata accumulate
exponential values, then inserting a barrier actually decreases the partials sum of values returned by slave automata (Lemma~\ref{l:barrierDecreasesPartialSum}), and hence we can insert barriers even at infinitely many positions and the value of the resulting word does not exceed the value of the original word.
Thanks to barriers, we can show that for every word $w$, there exists a word $w'$ such that at every position at most exponentially many slave automata aggregate 
exponential values and  the value of $w'$ does not exceed the value of $w$. 
Slave automata which aggregate bounded (exponential) values can be eliminated, i.e., we construct an NWA $\nestedA'$ which simulates $\nestedA'$ in such a way that 
runs of slave automata that accumulate at most exponential values are compressed into a single transition. Observe that $\nestedA'$ on word $w'$ as above has exponential width.
Hence, the infimum of $\nestedA$ over all words coincides with the infimum of $\nestedA'$ over words on which it has a run of exponential width.
We can encode the bound on width into $\nestedA'$ and decide the emptiness problem of $\nestedA'$ in non-deterministic logarithmic space (Theorem~\ref{th:BoundedBF}).
The size of $\nestedA'$ is doubly-exponentially bounded in the size of $\nestedA$, and hence the emptiness problem for finite-width $(\flimavg;\fsum^+)$-automata with bidirectional slave automata is in $\EXPSPACE$.
\smallskip

\Paragraph{Configuration and multiplicities}.
Invocation of a slave automaton in an NWA is a form of a universal transition 
in the sense of alternating automata. 
We adapt the power-set construction, which is used to convert alternating
automata to non-deterministic automata, to the NWA case. 
Given a (non-deterministic) NWA $\nestedA$ with bidirectional slave automata, we define \emph{configurations} and \emph{multiplicities} of $\nestedA$
as follows.
Let $\Qslv$ be the disjoint union of the sets of states of all slave automata of $\nestedA$.
For a run of $\nestedA$, we say that $(q_m, A)$ is the \emph{configuration} at position $p$ if
$q_m$ is the state of the master automaton at position $p$
and $A \subseteq \Qslv$ is the set of states of slave automata at position $p$.
We denote by $\confNum{\nestedA}$ the number of configurations of $\nestedA$.
We define the \emph{multiplicity} $\mult$ at position $p$ as the function
$\mult : \Qslv \mapsto \N$, such that $\mult(q)$ specifies the number of 
slave automata in the state $q$ at position $p$.
The configuration together with the multiplicity give a complete 
description of the state of $\nestedA$ at position $p$.
\smallskip

We observe that 
without loss of generality we can assume that NWA are deterministic. 
Basically, non-deterministic choices of the master automaton and slave automata can be encoded in the input alphabet. 
More precisely, the proof consists of two steps. 
First, we define simple runs as follows.
A run of an NWA is \emph{simple} if at every position in the run slave automata that are in the same state take the same transition.
We show that (i)~for every run $(\masterRun, \slaveRun_1, \slaveRun_2, \dots)$ of $\nestedA$ there exists
a \emph{simple} run of $\nestedA$ of the value not exceeding the value of $(\masterRun, \slaveRun_1, \slaveRun_2, \dots)$. 
Second, we show that (ii)~there exists a deterministic $(\flimavg; \fsum^+)$-automaton $\nestedA'$ over
an extended alphabet such that 
the sets of accepting simple runs of $\nestedA$ and accepting runs of $\nestedA'$
coincide and each run has the same value in both automata.
 
The proof of the following lemma  is virtually the same as in the case of NWA with forward-walking slave automata only~\cite{nested}.

\begin{restatable}{lemma}{WLOGDeterministicLimAvgSum}	
\label{WLOGDeterministicLimAvgSum}
Given a $(\flimavg; \fsum^+)$-automaton $\nestedA$ over $\Sigma$ with bidirectional slave automata,
(i)~for every run of $\nestedA$, there exists a simple run of at most the same value, and
(ii)~one can compute in polynomial space a deterministic $(\flimavg; \fsum^+)$-automaton $\nestedA'$
over an alphabet $\Sigma \times \Gamma$ such that:
(1)~$\inf_{w \in \Sigma^+} \valueL{\aut}(w) = \inf_{w' \in (\Sigma \times \Gamma)^+} \valueL{\aut'}(w')$, and
(2)~$\confNum{\nestedA} = \confNum{\nestedA'}$.
\end{restatable}
\begin{proof}
(i): Consider a run $(\masterRun, \slaveRun_1, \slaveRun_2, \dots)$ of $\nestedA$.
Suppose that $\pi_i, \pi_{j}$ are runs of the same slave automaton $\slaveA$ invoked at positions $i$ and $j$, such that its both instances are 
in the same state at position $s$ in the word, i.e., $\pi_i[i'] = \pi_j[j']$,  where
$i', j'$ are the positions in $\pi_i, \pi_j$ corresponding to the position $s$ in $w$.
We pick from the suffixes $\pi_i[i', |\pi_i|], \pi_j[j', |\pi_j|]$ the one with the smaller sum, and in case of the equal sum we pick the shorter.
Then, we change the suffixes of both runs to the picked one. 
Such a transformation does not increase the value of the partial sums and
does not introduce infinite runs of slave automata. 
Indeed, a run of each slave automaton can be changed by such an operation only finitely many times.
Thus, this transformation can be applied to any pair of slave runs to 
obtain a simple run of the value not exceeding the value of $(\masterRun, \slaveRun_1, \slaveRun_2, \dots)$.

\newcommand{\Qall}{Q_{\textrm{all}}}
(ii): 
Without loss of generality, we can assume that  for every slave automaton in $\nestedA$
final states have no outgoing transitions. 
Let $\Qall$ be the disjoint union of 
the sets of states of the master automaton and all slave automata of $\nestedA$.
We define $\Gamma$ as the set of all partial functions $h : \Qall \mapsto \Qall$.
We define a $(\flimavg; \fsum^+)$-automaton $\nestedA'$ over the alphabet $\Sigma \times \Gamma$ 
by modifying only the transition relations and labeling functions of the master automaton and slave automata of $\nestedA$; 
the sets of states and accepting states are the same as in the original automata.
The transition relation and the labeling function of the master automaton $\masterA'$ of
$\nestedA'$ are defined as follows: 
$(q, \lpair{a}{h},q')$ is a transition of $\masterA'$ if and only if $h(q) = q'$ and $\masterA$ has the transition $(q,a,q')$.
The label of the transition $(q, \lpair{a}{h},q')$ is the same as the label 
of the transition $(q,a,q')$ in $\masterA$.  Similarly, 
for each slave automaton $\slaveA_i$ in $\nestedA$, the transition relation and the labeling function of the corresponding
slave automaton $\slaveA_i'$ in $\nestedA'$ are defined as follows:
$(q, \lpair{a}{h},q')$ is a transition of  $\slaveA_i'$ if and only if $h(q) = q'$ and $\slaveA_i$ has the transition $(q,a,q')$.
The label of the transition $(q, \lpair{a}{h},q')$ is the same as the label 
of the transition $(q,a,q')$ in $\slaveA_i$. 

First, we see that $\confNum{\nestedA} = \confNum{\nestedA'}$.
Second, observe that the master automaton $\masterA'$ and all slave automata $\slaveA_i'$
are deterministic. Moreover, since we assumed that for every slave automaton in $\nestedA$
final states have no outgoing transitions, slave automata $\slaveA_i'$ recognize 
prefix free languages. 
Finally, it follows from the construction that
(i)~every simple run $(\masterRun, \slaveRun_1, \slaveRun_2, \dots)$ of $\nestedA$ is a
 run of $\nestedA'$ of the same value. Basically, we encode in the input word all transitions in functions $h \in \Gamma$.
The value of each transition is the same by the construction.
Conversely, (ii)~every run $(\masterRun, \slaveRun_1, \slaveRun_2, \dots)$ of $\nestedA'$ is a
simple run of $\nestedA$ of the same value. Indeed, the fact that transitions are
directed by functions $h \in \Gamma$ implies that the run is simple.
\end{proof}

In the following definition we introduce \emph{barriers}, which are words on which all active slave automata terminate, i.e., if $u$ is a barrier, then
forward-walking slave automata terminate while reading $u$, and backward-walking slave automata terminate while reading $u^R$ (word $u$ from right to left).
Barriers have additional properties, which allow us tho show that if exponentially many slave automata accumulate exponential values, 
then inserting a barrier decreases multiplicities of slave automata and does not increase the partial sums of values returned by slave automata.
		
\begin{definition}[Barriers]
\label{def:barriers}
Let $w$ be an infinite word, $i$ be a position and $k > i$ be the first position such that all backward-walking slave automata invoked past $k$ terminate past $i$.
A finite word $u$ is a barrier at $i$ in $w$ if in word $w' = w[1,i] u w[i+1, \infty]$ we have
\begin{enumerate}[BC1]
\item backward-walking slave automata invoked past position $i+|u|$ terminate past position $i$ (between positions $i+1$ and $i+|u|$) 
\item forward-walking slave automata invoked before position $i$  terminate before position $i+|u|$, 
\item the configurations at $i$ and $i + |u|$ in $w'$ are the same as the one at $i$ in $w$,  
\item the length of $u$ is bounded by $\boundOnMulti = (|Q_s + 2|) \cdot \confNum{\nestedA} \cdot |Q_s|^{2 \cdot |Q_s|}$,
\item 
for every state $q$ of some backward-walking (resp., forward-walking) slave automaton, 
the multiplicity $mult(q)$ at $i$ in $w'$ (resp., $\mult(q)$ at $i+|u|$ in $w'$) is bounded by $mult(q)$ at $i$ in $w$, and
\item every backward-walking (resp., forward-walking) slave automaton active at position $i$,
accumulates over $w[1,i]$ (resp., $w[i,\infty]$) a value greater or equal to the value accumulated over $w[1,i]u$ (resp., $uw[i+1,\infty]$).
\end{enumerate}
\end{definition}

The above conditions simply state that a barrier terminates all slave automata active and reduces their multiplicities, i.e.,
the multiplicity of backward-walking slave automata, which are invoked in the suffix $w[i+1,\infty]$ is reduced by word $u$ (BC1) and
the multiplicity of forward-walking slave automata invoked in prefix $w[1,i]$ is reduced by word $u$ (BC2). 
Property BC3 ensures
that inserting  a barrier at position $i$ does not change the run essentially (except for $u$ it only reduces the multiplicities of slave automata).
These properties and the bound on the length of barriers (BC4) allow us to reduce the multiplicity of slave automata along words. 
Properties BC5, BC6 are necessary to show that such a reduction of multiplicities does not increase the values of words.

\begin{lemma}
\label{l:barriersExist}
Let $\nestedA$ be a deterministic $(\flimavg; \fsum^+)$-automaton of finite width with 
 bidirectional slave automata. Then, for every word $w$,
 at almost every position a barrier exists. 
\end{lemma}
\begin{proof}
Let $w$ be a word. We consider the unique run of $\nestedA$ on $w$ and refer to the positions 
in $w$ and the corresponding positions in the run, i.e., we refer to ``the configuration at $i$ in  $w$'' as the unique configuration 
of $\nestedA$ at $i$ while processing word $w$.

We define the \emph{profile} at position $j$ in $w$ as a pair of   
the configuration at $j$ and multiplicities at $j$ bounded by $\boundOnMulti$, defined for every 
$q \in Q_s$ as $min(\mult(q), \boundOnMulti)$.
Let $c_0 < c_1$ positions in $w$  such that
every profile that occurs infinitely often in $w$ occurs between $c_0$ and $c_1$. 
Next, we define $c_2$ as the minimal position past $c_1$ such that every backward-walking slave
automaton invoked past $c_2$ terminates at some position past $c_1$, i.e., any backward-walking slave automaton invoked 
past  $c_2$ terminates before it reaches $c_1$. The NWA $\nestedA$ has finite width and hence such $c_2$ exists.   
We show that there exists a barrier in $w$ for every position $i > c_2$.

\Paragraph{Construction of a barrier $u$ at position $i$}.
Let $i > c_2$. We pick positions $a < i < b$ such that all
slave automata active at $i$ terminate within $w[a,b]$ and the profiles at positions $a,i,b$ and letters in $w$ are the same. 
Since $i > c_2 > c_1 > c_0$ such $a,b$ exist. 
Observe that $w[a,b]$ satisfies first three conditions of the barrier definition, but it does not have to satisfy the remaining conditions. 
Consider word $w[b,i]w[a+1,i]$. It satisfies all barrier conditions except for BC4.
We take $w[b,i]w[a+1,i]$ and transform it into a barrier $u$ by removing certain subwords corresponding to cycles in $\nestedA$, i.e.,
subwords such that at the beginning and at the end of this subword $\nestedA$ is in the same configuration.
Removal of such subwords does not change runs of the master automaton or slave automata in the suffix of the word.
However, to show condition RC5 we need to ensure that the removal operation does not change the profile, i.e., 
the profiles at positions $i$ and $i+|u|$ in $w[1,i] u w[i+1, \infty]$ are the same as the profile at $i$ in $w$.

We define an \emph{extended configuration} in a finite word $x$ at position $p$ as the pair of
the configuration at $p$ and the equivalence relation $R_p$  on states of slave automata active at position $p$
such that $q_1 R_p q_2$ if and only if either
\begin{compactitem}
\item $q_1, q_2$ are states of forward-walking slave automata  and 
slave automata in states $q_1$ and $q_2$ at position $p$ reach the same 
state at the end of word $x$, or
\item $q_1, q_2$ are states of backward-walking slave automata and 
slave automata in states $q_1$ and $q_2$ at position $p$ reach the same 
state at the beginning of word $x$.
\end{compactitem}
The slave automata that do not reach the end o word $x$ (resp., the beginning of $x$) are in the same equivalence class $\emptyset$.
Given two positions $p < p'$ in $x$ with the same extended configuration, we define
\emph{transformation from $p$ to $p'$} as
a function  from  the set of equivalence classes of $R_p$ (which is equal $R_{p'}$) into itself such that
\begin{compactitem}
\item for a state $q$ of a forward-walking slave automaton, the 
equivalence class $[q]$ is transformed into a class $[q']$ if some slave automaton in state $q$ at position $p$
reaches state from $[q']$ at position $p'$, and
\item for a state $q$ of a backward-walking slave automaton, the 
equivalence class $[q]$ is transformed into a class $[q']$ if some slave automaton in state $q'$ at position $p'$
reaches state from $[q]$ at position $p$. 
\end{compactitem}
Due to determinicity of $\nestedA$, the transformation is, indeed, a function and a permutation.

Now we describe the subword removal process. First, we mark 
positions $j_1, \ldots, j_n$  in $w$ at which each of
slave automata active at position $i$ in $w$ terminates. 
Observe that these belong to the interval $[a,b]$. 
Now, starting with word $w[b,i]w[a+1,i]$ we iteratively remove subwords $y$ such that
(a)~the extended configurations at the first and the last position of $y$ are the same,
(b)~the transformation between these positions  is the identity, and
(c)~$y$ does not contain any  position corresponding to positions $\{j_1, \ldots, j_{n}\} $.
The last word, from which no such a subword can be removed is $u$. We show that $u$ is a barrier.

\Paragraph{Word $u$ is a barrier}.
The positions at which slave automata are terminated are not removed from $w[a,b]$  and hence $u$
satisfies conditions BC1, BC2.
Removal of a word satisfying (a) and (b) preserves profile at every step and hence condition BC3 holds for $u$.

Condition BC4 follows from the fact that there are at most $\confNum{\nestedA} \cdot |Q_s|^{|Q_s|}$ 
different extended configurations. Transformations are permutations of equivalence classes, i.e., they are permutations of
sets of size at most  $|Q_s|$. Therefore, if $k > |Q_s|^{|Q_s|}$, then
 among positions $p_1 < \ldots < p_k$ with the same extended configuration there exists a pair such that 
the transformation between these positions  is the identity. 
Thus, words of length at least $\confNum{\nestedA} \cdot |Q_s|^{2 \cdot |Q_s|}$ contain a subword $y$ satisfying 
(a) and (b). Furthermore, words of length at least $(|Q_s + 2|) \cdot \confNum{\nestedA} \cdot |Q_s|^{2 \cdot |Q_s|} = 
{\boundOnMulti}$, contain a subword $y$ satisfying (a), (b) and (c). Therefore, the length of $v$ is bounded by ${\boundOnMulti}$.

We show that $u$ satisfies BC5.
Let $q$ be a state of some slave automaton and let $mult_1(q)$ be the multiplicity of $q$ in $w$ at position $i$
and $mult_2(q), mult_3(q)$ be multiplicities of $q$ in $w[1,i] u w[i+1, \infty]$ at positions $i$ and $i+|u|$ respectively.
 Observe that (a) and (b) imply the 
the profiles at positions $i$ and $i+|u|$ in $w[1,i] u w[i+1, \infty]$ are the same as the profile at $i$ in $w$.
It follows that $min(mult_1(q), N) = min(mult_2(q), N)  = min(mult_3(q), N)$. 
Now, if $q$ is a state of a backward-walking slave automaton, then all such automata active at position $i$
in  $w[1,i] u w[i+1, \infty]$ have been invoked in $v$ and hence $mult_2(q) < |u| \leq {\boundOnMulti}$.
It follows that  $mult_2(q) \leq mult_1(q)$. 
Otherwise, similarly if $q$ is a state of a forward-walking slave automaton we have $mult_3(q) < N$ and
$mult_3(q) \leq mult_1(q)$. Therefore, condition BC5 holds.

Finally, word $u$ satisfies condition BC6.
Consider a backward-walking slave automaton $\slaveA$ active at position $i$ in $w$.
This automaton terminates before position $a$ and hence it accumulates equal values on  subwords  
$w[a,i]$ and $w[1,i]w[b,i]w[a+1,i]$. Now, word $u$ results from  $w[b,i]w[a+1,i]$ by deletion of certain subwords, while it preserves positions corresponding to
termination of slave automata. Moreover, the deletion process only shortens runs; the transitions taken by slave automata correspond to the transitions in the original run.
Thus, $\slaveA$
accumulates over $w[1,i]u$ a value smaller or equal to the value accumulated over $w[1,i]w[b,i]w[a+1,i]$.
The case of forward-walking slave automata is symmetric.
\end{proof}

Now, we show the key property of barriers, i.e., they decrease the partial sum of values if inserted at a position where 
exponentially many slave automata accumulate exponential values. This enables us to reduce the emptiness problem for finite-width NWA to
the bounded-width case. 

\begin{lemma}
\label{l:barrierDecreasesPartialSum}
Let $\nestedA$ be a deterministic $(\flimavg;\fsum^+)$-automaton with bidirectional slave automata.
Let $C$ be the maximal weight of slave automata of $\nestedA$.
Let $w$ be a word and let $u$ be a barrier for $w$ at position $i$.
If more than $2 \cdot |\nestedA| \cdot \boundOnMulti$ slave automata accumulate
the value exceeding $4 \cdot C \cdot \boundOnMulti$ past $i$ in $w$ (i.e, for backward-walking slave automata it is at $w[1,i]$), 
then for almost all $K$, the sum of values returned by slave automata invoked up to $k$ in $w$ is greater than
the sum of values of  returned by slave automata invoked up to $k+|u|$ in $w[1,i]uw[i+1,\infty]$. 
\end{lemma}

\begin{proof}
Let $k$ be the first position past $i$ such that all backward-walking slave automata invoked past $k$ terminate before position $i$.
We show that the partial sum of values returned by slave automata invoked up to $k$ in $w$ is greater than
the partial sum of values of  returned by slave automata invoked up to $k+|u|$ in $w[1,i]uw[i+1,\infty]$. This argument works for every $k_0 \geq k$. 
The partial sum of values returned by slave automata invoked up to $k$ in $w$ consists of the sum of weights:
(1)~accumulated by slave automata before they reach position $i$, and
(2)~accumulated by slave automata after they reach position $i$.	
In (1) we include values of backward-walking (resp., forward-walking) slave automata that do not reach position $i$.
Let $A$ be the set of states of slave automata active at position $i$ in $w$. 
For $q \in A$, we define $\mult(q)$ (resp., $val(q)$) as the multiplicity at $i$ (resp., the value accumulated past $i$) in $w$ by slave automata that are in the state $q$ at $i$.
Observe that (2) = $\sum_{q \in A} mult(q) \cdot 	val(q)$. 

The partial sum of values returned by slave automata invoked up to $k+|u|$ in $w$ consists of four components, which are the sum of weights 
(1')~aggregated by slave automata before they reach any of positions $i, \ldots i+|u|$, 
(2'a)~aggregated over positions  $i, \ldots i+|u|$ by backward-walking slave automata invoked past $i+|u|$ and forward-walking slave automata invoked before $i$, 
(2'b)~aggregated past positions  $i, \ldots i+|u|$ by slave automata invoked on positions $i, \ldots i+|u|$, i.e.,
the values aggregated by backward-walking slave automata over $w'[1,i] = w[1,i]$ and 
forward-walking slave automata over $w'[i+|u|+1, \infty] = w[i+1,\infty]$, and
(2'c)~aggregated over positions  $i, \ldots i+|u|$ by slave automata invoked on positions $i, \ldots i+|u|$.

We observe that (1)$=$(1') and we show that (2) $>$ (2'a)+(2'b)+(2'c).
Observe that (2'c) is bounded by $C \cdot |u|^2$. In (2'a), the multiplicities of slave automata are the same as the multiplicities of 
the corresponding slave automata active at $i$ in $w$ while the values they accumulate are bounded by the minimum of 
the values accumulated within $w$ and $C \cdot |u|$. Indeed, consider a backward-walking slave automaton, which is in state $q$ at position $i+|u|$ in $w'$.
The multiplicity of such automata is equal to  the multiplicity of slave automata that are in state $q$ at position $i$ in $w$.
Moreover, due to condition BC6 satisfied by $u$, the value which such an automaton accumulates along $w'[1,i+|u|]$ is bounded by the value it accumulates at $w[1,i]$.
Also, such an automaton terminates within $|u|$ and hence its value is bounded by $C \cdot |u|$ as well.
The similar reasoning holds for forward-walking slave automata active at position $i$ in $w'$ and hence
(2'a) is bounded by $\sum_{a \in A} \mult(q) \min(C\cdot |u|, val(q))$.
In (2'b), slave automata accumulate the same values as the corresponding slave automata in $w$, but multiplicities are bounded by $|u|$ and the values of the corresponding
slave automata at position $i$ in $w$.
	Indeed, consider a backward-walking slave automaton, which is in state $q$ at position $i$ in $w'$.
Since the profile at $i$ in $w$ and $w'$ is the same, the multiplicity $\mult'(q)$ of such automata is bounded by the multiplicity of slave automata that are in state $q$ at position $i$ in $w$.
Moreover, all backward-walking slave automata active at position $i$ in $w'$ have been invoked within positions $i, \ldots, i + |u|$ and hence
the sum of their multiplicities is bounded by $|u|$.  
Since $w[1,i] = w'[1,i]$, a backward-walking slave automaton in state $q$ at position $i$ in $w'$ accumulates value $val(q)$ at past position $i$, i.e, 
the value which accumulates the same automaton in $w$.
Similar estimates hold for forward-walking slave automata past position $i+|u|$.
Therefore, (2'b) is bounded by $\sum_{q \in A} mult'(q)  val(q)$, where $\sum_{q \in A} mult'(q) \leq |u|$.
Due to condition BC5 of barriers, for every $q \in A$ we have $\mult'(q) \leq mult(q)$.
Now, we estimate $(2) - (2'a) - (2'b) - (2'	c)$, which equals
$\sum_{q \in A} (\mult(q) \cdot (val(q) -  \min(C\cdot |u|, val(q))) - \sum_{q \in A} mult'(q)  val(q) ) - C \cdot |u|^2$.
We partition $A$ into $A_1$,
the states of slave automata at position $i$, which accumulate the value at least $4 \cdot C \cdot \boundOnMulti$ past position $i$ in $w$, and $A_2$ the reaming states from $A$.
Then,
 (2) - (2'a)+(2'b)+(2'c) $\geq \sum_{q \in A_1} (\mult(q) - \mult'(q)) \cdot (val(q) - C\cdot |u|) -  \sum_{q\in A_1} \mult'(q) C \cdot |u|  - \sum_{q \in A_2} \mult'(q)  4 \cdot C \cdot \boundOnMulti - C \cdot |u|^2$.
Since $\sum_{q \in A} mult'(q) \leq |u|$ and  $|u| \leq \boundOnMulti$, we get
 (2) - (2'a)+(2'b)+(2'c) $ \geq \sum_{q \in A_1} (\mult(q) - \mult'(q)) \cdot (val(q) - C\cdot |u|) - 5 \cdot C \cdot \boundOnMulti^2$.
  Recall that $\sum_{q \in A_1} \mult(q) \geq 2 \cdot \boundOnMulti$ and for every $q \in A_1$ we have $val(q) \geq 4 \cdot C \cdot \boundOnMulti$. 
Therefore,  $\sum_{q \in A_1} (\mult(q) - \mult'(q)) \cdot (val(q) - C\cdot |u|) > 8 \cdot C \cdot \boundOnMulti^2$ and hence (2) $>$ (2'a)+(2'b)+(2'c).
This concludes the proof that insertion of a barrier decreases the partial sum.

\end{proof}

Using the above lemma we can reduce the emptiness problem for finite-width $(\flimavg;\fsum^+)$-automata with bidirectional slave automata
to the bounded-width case.

\begin{lemma}
\label{l:reductionFiniteWidthToBoundedWidth}
Let $\nestedA$ be a deterministic $(\flimavg; \fsum^+)$-automaton of finite width with 
bidirectional slave automata.
There exists a deterministic $(\flimavg; \fsum^+)$-automaton $\nestedA'$ over an extended alphabet $\Sigma_1$ 
with bidirectional slave automata of width bounded exponentially in $|\nestedA|$,
such that 
the emptiness problems for $\nestedA$ and $\nestedA'$ coincide, i.e.,
$\inf_{w \in \Sigma^{\omega}} \valueL{\nestedA}(w) = \inf_{w \in \Sigma_1^{\omega}} \valueL{\nestedA'}(w)$.
The size of $\nestedA'$ is $O(\boundOnMulti \cdot |\nestedA|)$.
\end{lemma}
\begin{proof}
We define an exponential size $(\flimavg; \fsum^+)$-automaton $\nestedA'$ with 
 bidirectional slave automata of width bounded by exponentially in $|\nestedA|$ such that $\nestedA'$
  simulates a subset of runes of  $\nestedA$ which satisfy the following condition (*):
  at almost every position $s$,
among slave automata active at $s$, at most $2 \cdot |\nestedA| \cdot \boundOnMulti$ will accumulate value greater than $4 \cdot C \cdot \boundOnMulti$.
Next, we show that for every run of $\nestedA$ there exists a run satisfying (*) of at most the same value. It follows
that the emptiness problems for $\nestedA$ and $\nestedA'$ coincide.

\Paragraph{Definition of $\nestedA$}. 
The alphabet $\Sigma_1$ consists of $\Sigma$ and additional marking letters described below.
First, the automaton $\nestedA'$ invokes only dummy slave automata before the initiation marking. 
Past initial marking the master automaton keeps track of the number of invoked slave automata and rejects
if the number of slave automata exceeds $2 \cdot |\nestedA| \cdot \boundOnMulti$.
Second, we modify each slave automaton of $\nestedA$ so that they work only as long as they can accumulate the value
exceeding $4\cdot C \cdot \boundOnMulti$, where $C$ is the maximal weight among all slave automata of $\nestedA$.
In particular, slave automata, which accumulate the value below  $4\cdot C \cdot \boundOnMulti$, are invoked for a single transition only.
We encode in the alphabet whether a slave automaton is invoked for a single transition and the weight of this transition.
Moreover, we encode in the alphabet that slave automata in state $q$  ought to terminate as in the following run they
accumulate the value below  $4\cdot C \cdot \boundOnMulti$. The master automaton checks whether the external markings are correct. 
These constructions involve a single exponential blow up of the master automaton, slave automata and the alphabet.
Observe that accepting runs of $\nestedA'$ correspond to accepting runs of $\nestedA$, which satisfy condition (*).
The values of the corresponding runs are the same. Now, we need to show that while computing the infimum over all accepting runs of $\nestedA$, we can restrict ourselves to runs satisfying (*).

\Paragraph{The emptiness problems for $\nestedA$ and $\nestedA'$ coincide}.
Since $\nestedA'$ simulates a subset of runs of $\nestedA$ we have 
$\inf_{w \in \Sigma^{\omega}} \valueL{\nestedA}(w) \leq \inf_{w \in \Sigma_1^{\omega}} \valueL{\nestedA'}(w)$.
We show how to transform any word $w$ of $\nestedA$ into a word whose unique run satisfies (*), which means that is can be simulated by $\nestedA'$, and whose value does not exceed the value of $w$.
It follows that $\inf_{w \in \Sigma^{\omega}} \valueL{\nestedA}(w) \geq \inf_{w \in \Sigma_1^{\omega}} \valueL{\nestedA'}(w)$.

Let $w$ be a word accepted by $\nestedA$.
Let $s$ be a position such that on every position $p\geq s$ in $w$ a barrier exists (Lemma~\ref{l:barriersExist}).
We start with $i=0$, word $w_0 = w$ and $p_0 = s$. We iteratively at step $i$, pick first position $p > p_i$ at which more than $2 \cdot |\nestedA| \cdot \boundOnMulti$ slave automata accumulate
the values exceeding $4 \cdot C \cdot \boundOnMulti$ and we insert a barrier $u_i$ at position $p$.
Barrier $u_i$  exists as a barrier for $w$ at the corresponding position. 
Then, we start the next iteration $i+1$ with word $w_{i+1} = w_i[1,p] u_i w_i[p+1, \infty]$ and $p_{i+1} = p +|u_i|$.
Observe that  iterations past $i$ change only positions past $p_i$ in words $w_i, w_{i+1}, \ldots$, i.e., all words $w_i, w_{i+1}, \ldots$ share the prefix $w[1,p_i]$.
Thus, there exists the limit word $w_B$, which is the result of the iterative process.  
 We argue that the resulting word $w_B$ satisfies (*)
and its value does not exceed the value of $w$.
First, we show that $w_B$ satisfies (*).

Let $u_i$ be a barrier at position $p$ in $w_i$ and let $w' = w_i[1,p] u_i w_i[p+1, \infty]$.
The forward-walking slave automata active at $p$ in $w_i$ are terminated within $u_i$ in $w'$ and hence accumulate the value below $|u_i| < C \cdot N$ in $w'$, where $C$ is the maximal weight of slave automata of $\nestedA$. 
The number of forward-walking slave automata active between positions $p$ and $p+|u_i|$ which are not terminated within $u$ is bounded 
by $|u_i| < \boundOnMulti$. Therefore, at every position between $i$ and $p+|u|$ in $w'$, 
at most $\boundOnMulti$ forward-walking slave automata accumulate the value exceeding $4\cdot C \cdot \boundOnMulti$.
The similar argument applies to backward-walking slave automata, which shows that at positions $p$ to $p+|u|$ condition (*) is satisfied. 

It remains to comment on the value of $w_B$. Barriers inserted in the iterative process satisfy conditions of Lemma~\ref{l:barrierDecreasesPartialSum}, and hence
the partial sums decrease from $w_0$ to $w_1, w_2$ and so on. It follows partial sums in $w_B$ are bounded by partial sums in every word $w_0, w_1, \ldots$ and hence 
the value of $w_B$ does not exceed the value of $w$. 
\end{proof}

\Paragraph{Algorithm}.
Let $\nestedA$ be a non-deterministic finite-width $(\flimavg;\fsum^+)$-automaton with bidirectional slave automata.
Lemma~\ref{WLOGDeterministicLimAvgSum} reduces the emptiness problem of $\nestedA$ to
the emptiness problem of $\nestedA_d$, which has the same properties as $\nestedA$, but it is deterministic.
The reduction takes exponential time and produces an exponential-size automaton, while it does not change the number of configurations. 
Therefore, 
the value of $\boundOnMulti$ for $\nestedA$ and $\nestedA_d$ is the same.
Next, Lemma~\ref{l:reductionFiniteWidthToBoundedWidth} reduces the emptiness problem for $\nestedA_d$ to
the emptiness problem of $\nestedA_B$ which is a deterministic  $(\flimavg;\fsum^+)$-automaton with bidirectional slave automata and
the width of $\nestedA_B$ is linear in $\boundOnMulti$, i.e., it is exponential in $|\nestedA|$. 
The size of $\nestedA_B$ is $O(\boundOnMulti \cdot |\nestedA_d|)$, i.e., it is exponential in $|\nestedA|$.
Therefore, the emptiness problem for $\nestedA_B$ can be solved in polynomial space in $\boundOnMulti \cdot |\nestedA_B|$ (Theorem~\ref{th:BoundedBF}) and 
in the exponential space in the size of $\nestedA$.

\subsection{The proof of Theorem~\ref{th:expressivness}}

In the finite-automata framework, the pumping lemma is a standard tool to show inexpressibility results. 
It is difficult to state a pumping lemma for NWA as their state space is infinite.
While there is finitely many configurations of NWA, the multiplicity of running slave automata is unbounded.
We consider the graph of configurations of an NWA to reduce the infinite-space case to the finite-state case. 

Recall that we assume that slave automata do not have outgoing transitions from accepting states. We also assume (without loss of generality) that initial states of slave automata 
do not have ingoing transitions.

\Paragraph{Graph of configurations}.
Let $\nestedA$ be an NWA with bidirectional slave automata over an alphabet $\Sigma$. 
We define the graph of configurations $G$ of an NWA $\nestedA$ as a $\Sigma$-labeled graph whose nodes are configurations of $\nestedA$.
For $a \in \Sigma$, there exists an $a$-labeled edge from configuration $(q_m, A_1)$ to $(q_m', A_2)$ if and only if 
\begin{compactitem}
\item the master automaton of $\nestedA$ has a transition $(q_m,a,q_m')$ invoking a slave automaton $\slaveA$ in an initial state $q_I$,
\item for $i=1,2$, states $A_i$  are partitioned into states of backward-walking slave automata $A_i^B$ and forward-walking slave automata $A_i^F$,
\item there exists a function $h_F$, which transforms all non-final states from $A_1^F$ onto $A_2^F \setminus \{q_I\}$ such that
for every $q \in dom(h_F)$ tuple $(q,a,h(q))$ is  a transition of some slave automaton,
\item there exists a function $h_B$, which transforms all non-final states from $A_2^B$ onto $A_1^B \setminus \{q_I\}$ such that
for every $q \in dom(h_B)$ tuple $(q,a,h(q))$ is  a transition of some slave automaton, and
\item if $\slaveA$ is forward-walking $q_I \in A_2^F$; otherwise $q_I \in A_1^B$.
\end{compactitem}

Paths in $G$ are related to simple runs of $\nestedA$.
Consider a simple run of $\nestedA$. Its sequence of configurations is an infinite path in $G$. 
Conversely, given a path $\pi$ in $G$ starting in the initial configuration $(q_I, \emptyset)$, we can specify regular conditions ensuring
(a)~existence of a simple run corresponding to $\pi$ and 
(b)~existence of a simple accepting run corresponding to $\pi$.
Regularity of these conditions follows from the fact that unweighted parts of runs of $\nestedA$ can be simulated by an alternating \buchi{} automaton.
This automaton checks whether runs of all slave automata are finite; this suffices to ensure that a path corresponds to a valid simple run. 
It can also check acceptance condition, i.e., whether runs of all slave automata terminate in accepting states and the master automaton visits some accepting state infinitely often. 

The graph of configurations of $\nestedA$ enables us to construct accepting runs of $\nestedA$ with desired properties. 
Having an accepting run of $\nestedA$ with a sequence of configurations $\alpha \beta \gamma$ 
such that $\beta$ is a cycle in the graph of configurations of $\nestedA$, we know that for every $n>0$ there exists an accepting run of $\nestedA$ whose sequence of configurations is 
$\alpha \beta^n \gamma$. This is a key observation used in the following lemma.

\IncomparableForwardAndBackward*
\begin{proof}
\newcommand{\confSeq}[1]{\sigma[#1]}
\newcommand{\bigO}{\mathcal{O}}
We show (i) in detail and next we comment on (ii). 
Suppose that $\nestedA$ is a $(\flimavg;\fsum^+)$-automaton with forward-walking slave automata which computes DCP (Example~\ref{ex:data-consistency}) restricted to the alphabet $\{r,\#,c\}$.
First, we show that in the graph of configurations of $\nestedA$ there exist two cycles $\tau_r, \tau_\#$ such that
$\tau_r$ is an $r$-labeled cycle in which at least one (non-dummy) slave automaton is invoked, and
$\tau_{\#}$ is an $\#$-labeled cycle in which $\nestedA$ takes only silent transitions, i.e., it invokes only slave automata that immediately accept returning no value. 
Next, we use $\tau_r, \tau_{\#}$ to construct of a run 
$\pi_0$ on some word $u$ such that the value of $\pi_0$ is smaller than DCP of $u$, which contradicts 
the assumption that $\nestedA$ computes DCP.

\Paragraph{Existence of $\tau_r$ and $\tau_{\#}$}.
Let $K$ be greater than the number of configurations of $\nestedA$ and let $N > 15 K^2$.
To simplify the calculations, we denote by $\bigO(2K)$ some natural number from interval $[0,2K]$. 
For example, we write $N + \frac{3K-1}{2} = N + \bigO(2K)$.
\smallskip

Consider a word $w = (c \#^N  r^{2K} c \#^{2N} r^K )^{\omega}$. DCP of $w$ is  $\frac{4}{3}\cdot N+ \bigO(2K)$. 
Let $\pi$ be a run of $\nestedA$ on $w$ of the value $\frac{4}{3}\cdot N + \bigO(2K)$. 
We can assume that $\pi$ is simple (Lemma~\ref{WLOGDeterministicLimAvgSum}).
In every block $\#^{2K}$ in $w$, there exist positions $i_1 < i_2$ such that
the configurations in $\pi$ at $i_1$ and $i_2$ are the same and $i_2 - i_1 > K$.	We remove these parts of $\pi$.
The resulting sequence $\pi'$ is a run of $\nestedA$ on some word 
$w' = c \#^N  r^{L_1} c \#^{2N} r^K c \#^N  r^{L_2} c \#^{2N} r^K \ldots$ such that $L_1, L_2, \ldots$ are at most $K-1$.
Observe that DCP of $w'$ is at least $\frac{3}{2}\cdot N$ and hence the value of $\pi'$ is at least $\frac{3}{2}\cdot N$.
However, the partial sums of the values returned by slave automata in $\pi'$ are bounded by
the corresponding partial sums in $\pi$. Therefore, the value of $\pi'$ increases due to the fact that the removed parts of $\pi$
contain invocations of slave automata returning small values and removal of these parts of $\pi$ increase partial averages.
It follows that infinitely often at least one (non-dummy) slave automaton is invoked over the block $r^{2K}$.
Consider the sequence of configurations $\confSeq{\pi}$ of run $\pi$. 
There exists infinitely many subsequences $\tau$ of $\confSeq{\pi}$, which correspond to 
transitions over letters $r$ and satisfy: (A1)~the first and the last configuration of $\tau$
is the same, (A2)~along $\tau$ at least one slave automaton is invoked and (A3)~the length of $\tau$ is bounded by $K$.
Such a sequence corresponds to an $r$-labeled cycle in the graph of configurations of $\nestedA$ of length at most $K$.
There are finitely many such cycles and hence there exists a cycle $\tau_{r}$ satisfying condition (A1), (A2) and (A3)
which occurs infinitely often in  $\confSeq{\pi}$.
 
In a similar we show that  $\confSeq{\pi}$ contains infinitely often a subsequence $\tau_{\#}$, which corresponds to transitions over letters $\#$,
 such that 
(B1)~the first and the last configuration of $\tau_{\#}$ is the same, 
(B2)~slave automata invoked along $\tau_{\#}$ return no value (correspond to silent transitions), and 
(B3)~the length of $\tau_{\#}$ is bounded by $K$.
To see that, we divide runs $\pi$ and $\pi'$ into blocks separated by letter $c$.
In transformation from $\pi$ to $\pi'$,
the average number of invocations of (non-dummy) slave automata per block decreases by at most $\frac{2K}{3}$.
Yet, the value of $\pi'$ increases by at least $\frac{1}{6}N - \bigO(2K)$ w.r.t. the value of $\pi$. 
Therefore, the average number of invoked slave automata per block
cannot exceed $9 K$ in $\pi$. 
It follows that  at most $14 \cdot K$ non-dummy slave automata are invoked on average in a block of $N$ letters $\#$.
Thus, there exists infinitely many occurrences of subsequences of $\confSeq{\pi}$  satisfying (B1), (B2) and (B3), and hence
there exists a $\#$-labeled cycle $\tau_{\#}$ satisfying conditions (B1), (B2) and (B3), which occurs infinitely often in $\confSeq{\pi}$.

Observe that there exist infinitely many subwords $c \#^N  r^{2K} c \#^{2N} r^K$ of $w$ such that in the corresponding positions in $\confSeq{\pi}$ occur both 
 $\tau_{\#}$ and $\tau_{r}$. Thus, there exists a path $\alpha_A$
in the graph of configurations of $\nestedA$ such that $\alpha_A$ leads from the last configuration of $\tau_{\#}$ to the first configuration of $\tau_r$ over letters $\#, r$.
Moreover, all slave automata in $\pi$ terminate after finite number of steps, while $\tau_{\#}$ and $\tau_r$ occur infinitely often. Therefore,
there exists a path $\alpha_B$  from the first configuration of $\tau_{r}$ to the last configuration of $\tau_{\#}$ over letters $\#,r,c$ such that 
 (C1)~at least one transition is over letter $c$, and 
 (C2)~all slave automata active at the first configuration of $\alpha_B$ are terminated before the end of $\alpha_B$,
(C3)~the master automaton of $\nestedA$ visits an accepting state within $\alpha_B$.
Let $u_A$ (resp., $u_B$) be a subword of $w$ at which configurations of $\pi$ form the sequence $\alpha_A$ (resp., $\alpha_B$).  
Next, we show the construction of $\pi_0$ using $\tau_r, \tau_{\#}$, $\alpha_A$ and $\alpha_B$.
\smallskip

\Paragraph{The construction of $\pi_0$}.
Let $M,L$ be natural numbers, which we fix later. 
We define $\pi_0$ as some simple accepting run that corresponds to the sequence of configurations $\alpha_0 ((\tau_{\#})^L \alpha_A (\tau_r)^M \alpha_B)^{\omega}$, 
where $\alpha_0$ is a sequence of configurations from an initial configuration to the first configuration of $\tau_{\#}$. 
Such a run exists as we can ensure that at positions corresponding to $\alpha_B$ all slave automata terminate in accepting states and the master automaton visits an accepting state.
Let $u_0$ be a word at which there exists a run with the sequence of configurations $\alpha_0$.
We define $u = u_0 (\#^{L\cdot |\tau_{\#}|} u_A r^{M \cdot |\tau_r|} u_B )^{\omega}$. 
The run $\pi_0$ is an accepting run on $u$.
Observe that DCP of $u$ exceeds $|\tau_{\#}| \cdot L$.
However, we show that the value of $\pi_0$ is smaller. 
Run $\pi_0$ is a lasso and its limit average is the average of the cycle, which corresponds to 
the average of $\alpha_B (\tau_{\#})^L \alpha_A (\tau_r)^M \alpha_B$ excluding values of slave automata invoked in the second occurrence of $\alpha_B$.
The non-dummy slave automata are invoked only in $\alpha_B$ and in  $\alpha_A (\tau_r)^M$. 
All slave automata invoked within this cycle terminate by the end of it, and hence 
(a)~the values of slave automata invoked in $\alpha_B$ are bounded by the length of the cycle multiplied by $C$, the maximal weight of $\nestedA$, i.e.,
$S_1 = C \cdot (|\alpha_B| + |\alpha_A| + L\cdot |\tau_{\#}| + M \cdot |\tau_r|)$, and
(b)~the values of slave automata invoked in $\alpha_A (\tau_r)^M$ are bounded by 
$S_2 = C \cdot (|\alpha_A| + |\alpha_B| + M \cdot |\tau_r|)$.
We have $S_1 > S_2$, however there are at most $|\alpha_B|$  slave automata invoked in $\alpha_B$, which accumulate value at most $S_1$.
The remaining slave automata are invoked in $\alpha_A (\tau_r)^M$ and there are at least $M$ of them.
Thus, the average value of the cycle is at most $ \frac{S_1 \cdot |\alpha_B| + S_2 \cdot M}{|\alpha_B| + M}$.
Now, for $M = 2 \cdot C \cdot |\alpha_B| \cdot |\tau_{\#}|$, we have
$\frac{S_1 \cdot |\alpha_B| }{|\alpha_B| + M} < |\alpha_B| + |\alpha_A| + \frac{L}{2} + C\cdot |\tau_r|\cdot |\alpha_B|$ and 
$\frac{S_2 \cdot M}{|\alpha_B| + M} < S_2$.
Let $L > 2 \cdot (S_2 +   |\alpha_B| + |\alpha_A| + C\cdot |\tau_r|\cdot |\alpha_B|)$, then the average of the cycle, which is bounded by 
$\frac{S_1 \cdot |\alpha_B|+ S_2 \cdot M}{M+ |\alpha_B|} $, is smaller than $L$. 
However, DCP of $u$ exceeds $L$, which contradicts the fact that $\nestedA$ computes DCP.

\Paragraph{Backward-walking slave automata}.
The proof for backward-walking slave automata is similar. 
We consider numbers $K,N$ and a word $w = (c w^{2K} \#^N  c w^K \#^{2N})^{\omega}$; we show that there exist cycles 
$\tau_w, \tau_{\#}'$ in $\nestedA$, with similar properties to $\tau_r, \tau_{\#}$ from the forward case.
Moreover, there exist sequences of configurations $\alpha_A'$ from the last configuration of $\tau_w$ to the first configuration of $\tau_{\#}'$, and
$\alpha_B'$ from the last configuration of $\tau_{\#}'$ to the first configuration of $\tau_{w}$, with the properties similar to (C1), (C2) and (C3).
To show (C2) we use the fact that $\tau_w$ and $\tau_{\#}'$ occur infinitely often and $\nestedA$ has \textbf{finite-width}, and hence for every position 
$i$ there exists position $j>i$ such that every (backward-walking) slave automaton active at position $j$ terminates before $i$ (i.e., at some position within $[i,j]$).
Next, we construct  from $\tau_w,\tau_{\#}',\alpha_A', \alpha_B'$ a run $\pi_0'$ of the value lower than DCP of the corresponding word.
The construction is virtually the same as in the forward case.
\end{proof} 
 
\section{Proofs from Section~\ref{s:bounded}}

\TechnicalBoundedWidth*
\begin{proof}
\Paragraph{(1)}: Assume that (*) is satisfied. 
Consider an accepting run  $\pi$ with configuration $\cycle[1]$ occurring infinitely often.
Let $i$ be a position at which configuration $\cycle[1]$ occurs. 
Let $i' < i$ be the last position at which any automaton from $Fc$ is invoked. 
Consider the run resulting from inserting cycle $\cycle$ repeated $N$ times at position $i$ in $\pi$.
The only slave automata active past position $i + |\cycle|$, which has been invoked before $i'$ are the automata from $Fc$.
Therefore, the partial sum of values returned by slave automata up to position $i'$ decreases by at least 
$(N-1) \cdot \gain(\cycle, Fc) + C < - (N-1) +C$, where $C$ is the value of forward slave automata invoked before $i'$, which terminate in $\cycle$.	
The number of slave automata invoked before $i'$ does not change and hence 
by picking $N$ large enough we can decease the partial average up to $i'$ arbitrarily.
We can apply such a pumping step at every position with configuration $\cycle[1]$ obtaining a run whose limit infimum of partial averages diverges to $-\infty$.
	 
\Paragraph{(2)}: $(\Rightarrow)$: 
Assume that there exists a cycle $\cycle$ in the graph of configurations of $\nestedA$ and 
a restriction $R$ such that $\AvgE(\cycle, R)$ and 
there exists an accepting run $\pi$ with configuration $\cycle[1]$ occurring infinitely often. 
Let $i$ be a position at which configuration $\cycle[1]$ occurs.
We insert $\cycle$ at position $i$ and obtain run $\pi'$. 
Let $i'$ be the last position in $\pi'$ such that $i' \geq i + |\cycle|$ and all automata active at position $i$, which are not in $R$, are invoked before $i'$. 
Due to presence of backward-walking slave automata $i'$ can be strictly grater than $i+ |\cycle|$.
Consider the run resulting from inserting cycle $\cycle$ repeated $N$ times at position $i$ in $\pi'$.
Then, the partial average up to position $i' + N |\cycle|$ is given by the expression
$\frac{a + N\cdot p - \Delta}{b + N \cdot q}$, where 
\begin{compactitem}
\item $a$ is the partial sum of values returned by slave automata invoked up to $i'$ in run $\pi$,
\item $b$ is the number of slave automata invoked up to $i'$ in run $\pi$,
\item $\frac{p}{q} = \AvgE(\cycle, R)$ and $q$ is the number of slave automata invoked in $\cycle$, and
\item  $\Delta$ is the value accumulated by backward-walking slave automata invoked past $i'+N |\cycle|$, which terminate  
within interval $[i + N |\cycle|, i + (N+1) |\cycle|]$.
\end{compactitem}
Now, by taking $N$ large enough we can bring the partial average arbitrarily close to $\frac{a}{b}$.
Using that and simple iteration, we can construct an accepting run of the value $\AvgE(\cycle, R)$.

$(\Leftarrow)$: Assume that $\nestedA$ does not satisfy (*) from (1). 
It follows that in every accepting run of $\nestedA$ for almost every position $i$,
 slave automata invoked before $i$, accumulate past $i$ the value exceeding value $D$
defined as $-C \cdot 	k^2 \cdot \conf{\nestedA}$, where 
$C$ is the maximal weight occurring in $\nestedA$ and $\conf{\nestedA}$ 
is the number of configurations of $\nestedA$. 
Indeed, if there exists such a position $i$, there exists a cycle past $i$ which we can pump to lower the sum of value returned by slave automata invoked before $i$.
Hence, existence of infinitely many such positions $i$, implies that condition (*) holds.

Let $\pi$ be an accepting run of $\nestedA$ of value $\lambda$.
Consider $\epsilon > 0$. There exists a prefix of $\pi$ up to position $i$ such that 
the partial average of values returned by slave automata up to $i$ is at most $\lambda + \epsilon$, 
the sum of values accumulated by slave automata invoked before $i$ exceeds $D$ and
 and the number of slave automata before $i$ exceeds $\frac{\epsilon}{-D}$.
Then, the partial average of the values accumulated by slave automata invoked before $i$ within positions $1, \ldots, i$
is at most $\lambda + 2\cdot \epsilon$.

Now, we can decompose the prefix up to $i$ into simple cycles one by one, i.e., having a prefix $\tau$ of $\pi$ up to position $i$, 
we pick a simple cycle, remove it from $\tau$ and repeat the process. We terminate when we end up with run $\tau_E$ which has no simple cycle to remove;
the remaining run $\tau_E$ has length bounded by the number of configurations and therefore its sum of values is greater than $D$.
Thus, the partial average of the values accumulated by slave automata invoked before $i$ within positions $1, \ldots, i$ equals (a)~
the weighted average of average weights of simple cycles excluding backward-walking slave automata invoked past $i$,
 plus (b)~the average of $\tau_E$ bounded by $D \cdot \frac{\epsilon}{D}$.
It follows that there exist a simple cycle $\cycle$ and a restriction $R$ such that $\AvgE(\cycle, R) \leq \lambda + 3 \cdot \epsilon$; 
otherwise the weighted average in (a) exceeds  $\lambda + 3 \cdot \epsilon$, which contradicts the choice of prefix of $\pi$.

However, there are infinitely many $\epsilon >0$, while there are finitely many simple cycles. Therefore, 
there exist a simple cycle $\cycle$ and a restriction $R$ such that $\AvgE(\cycle, R) \leq \lambda$. 
\end{proof}

\BoundedWidthForwardAndBackward*
\begin{proof}
It suffices to show that conditions from Lemma~\ref{l:techlical-bounded} can be checked (a)~in logarithmic space for constant $k$ and unary weights, 
(b)~polynomial time for constant $k$ and binary weights, 
 and (c)~polynomial space for $k$ given in unary.
Observe that these conditions reduce to weighted reachability, which can be computed in logarithmic space in the size of the graph of $k$-configurations of $\nestedA$, provided that weights can fit in logarithmic space.
Otherwise, if weights are represented in binary and are of length greater than logarithmic in the size of the graph, weighted reachability can be implemented using Dijkstra algorithm in polynomial time.
The size of the graph of $k$-configurations is polynomial in the size of $\nestedA$ and exponential in $k$. Thus,
 the graph of $k$-configurations of $\nestedA$ is polynomial if $k$ is constant, and exponential if $k$ is given in unary. 
 Finally, we comment how to compute the successor relation. 

We define a presuccessor relation $R$ on $k$-configurations as follows. 
We have $(q; q_1, \ldots, q_k) R (q'; q_1', \ldots, q_k')$ if and only if for some $a \in \Sigma$
the master automaton of $\nestedA$ has a transition $(q,a,q')$ invoking a slave automaton $\slaveA$ in an initial state $q_I$, and 
for every component $j \in \{1, \ldots, k\}$ one of the following holds
\begin{compactitem}
\item $q_j$ is a non-final state of a forward-walking slave automaton and $(q_j,a,q_j')$ is a transition of this automaton,
\item $q_j$ is a final state of a forward-walking slave automaton, and $q_j' = \bot$ or $q_j' = q_I$,
\item $q_j'$ is a non-final state of a backward-walking slave automaton and $(q_j',a,q_j)$ is a transition of this automaton,
\item $q_j'$ is a final state of a backward-walking slave automaton, and $q_j = \bot$ or $q_j = q_I$,
\end{compactitem}
if $\slaveA$ is forward-walking (resp., backward-walking) slave automaton, then for exactly one component $j$ we have $q_I  = q_j'$ (resp., $q_I = q_j$).
Observe that $R$ encodes a local consistency of transitions of the master and slave automata. A sequence of $k$-configurations consistent with $R$ satisfying
the following conditions (a) and (b) corresponds to an accepting simple run. These conditions are: 
(a)~the master automaton visits one of its accepting states infinitely often, and (b)~every slave automaton terminates after finitely many steps.
Now observe that the successor relation defined in Section~\ref{s:bounded} is the presuccessor relation restricted to $k$-configurations $C$, which are 
(1)~reachable through $R$ from the initial configuration and (2)~a cycle w.r.t. $R$ satisfying conditions (a) and (b) is reachable through $R$ from $C$. 
Indeed, for such configurations $C_1, C_2$ satisfying $C_1 R C_2$ there exists a sequence of $k$-configurations consistent with $R$, which corresponds to an accepting run.
It follows that the successor relation in the graph of $k$-configurations can be computed based on reachability w.r.t. presuccessor relation, which is computable in logarithmic space.
 \end{proof}

\end{document}